\documentclass[journal]{IEEEtran}
\usepackage[T1]{fontenc}
\usepackage{cite}
\ifCLASSINFOpdf
  \usepackage[pdftex]{graphicx}
\else
  \usepackage[dvips]{graphicx}
\fi
\usepackage{amssymb}
\usepackage{algorithm}
\usepackage{algorithmic}
\usepackage{tabularx}
\usepackage{amsthm,amssymb,amsfonts}
\usepackage{epsfig}
\usepackage{subfigure}
\usepackage{array}
\usepackage{changes}
\usepackage{textcomp}
\usepackage{xcolor}
\usepackage{amsmath}
\usepackage{fancyhdr}

\usepackage{dsfont}
\usepackage{amsmath,amssymb,amsfonts}
\usepackage{subfigure}
\usepackage{textcomp}
\usepackage{xcolor}
\usepackage{multirow}
\usepackage{mathtools}
\usepackage{mhchem}
\usepackage{tensor}
\usepackage{graphicx}  
\usepackage{psfrag}    
\usepackage{amsthm}
\usepackage{mathrsfs}
\usepackage{epstopdf}
\usepackage{bm}
\usepackage{cuted}
\usepackage{color, soul, framed}

\usepackage{xcolor}
\newtheorem{proposition}{Proposition}
\usepackage{xpatch}
\usepackage{tabularx}
\ifCLASSOPTIONcompsoc
  \usepackage[caption=false,font=normalsize,labelfont=sf,textfont=sf]{subfig}
\else
  \usepackage[caption=false,font=footnotesize]{subfig}
\fi
\usepackage{fixltx2e}
\usepackage{dblfloatfix}

\usepackage{multirow}
\usepackage{booktabs}
\usepackage{url}
\begin{document}

\title{Joint Trajectory and Resource Optimization for HAPs-SAR Systems with Energy-Aware Constraints
}

\author{Bang~Huang,~\IEEEmembership{Member,~IEEE,}	
Kihong Park,~\IEEEmembership{Senior Member,~IEEE,} 
Xiaowei Pang,
	Mohamed-Slim Alouini,~\IEEEmembership{Fellow,~IEEE,}
	\thanks{ The authors are with the Computer, Electrical and Mathematical Science and Engineering (CEMSE) division in King Abdullah University of Science and Technology (KAUST), Thuwal 6900, Makkah Province, Saudi Arabia.  (Emails: bang.huang@kaust.edu.sa; kihong.park@kaust.edu.sa; xiaowei.pang@kaust.edu.sa; slim.alouini@kaust.edu.sa;) (Corresponding author: Bang Huang)}}
\maketitle

\begin{abstract}
This paper investigates the joint optimization of trajectory planning and resource allocation for a high-altitude platform stations synthetic aperture radar (HAPs-SAR) system. To support real-time sensing and conserve the limited energy budget of the HAPs, the proposed framework assumes that the acquired radar data are transmitted in real time to a ground base station for SAR image reconstruction. A dynamic trajectory model is developed, and the power consumption associated with radar sensing, data transmission, and circular flight is comprehensively analyzed. In addition, solar energy harvesting is considered to enhance system sustainability. An energy-aware mixed-integer nonlinear programming (MINLP) problem is formulated to maximize radar beam coverage while satisfying operational constraints. To solve this challenging problem, a sub-optimal successive convex approximation (SCA)-based framework is proposed, incorporating iterative optimization and finite search. Simulation results validate the convergence of the proposed algorithm and demonstrate its effectiveness in balancing SAR performance, communication reliability, and energy efficiency. A final SAR imaging simulation on a 9-target lattice scenario further confirms the practical feasibility of the proposed solution.
\end{abstract}

\begin{IEEEkeywords}
Communication, High-altitude platform stations (HAPs), mixed-integer nonlinear programming (MINLP), SAR, successive convex approximations (SCA).
\end{IEEEkeywords}

\IEEEpeerreviewmaketitle

\section{Introduction}

\IEEEPARstart{t}{he}  \textit{Top 10 Emerging Technologies of 2024} report released by the \textit{World Economic Forum (WEF)} highlights the transformative potential of {high-altitude platform stations (HAPs)} in extending wireless connectivity to regions underserved by conventional terrestrial and satellite infrastructure~\cite{wef2024top10,belmekki2024cellular}. As a key enabler of {non-terrestrial networks (NTNs)}~\cite{anicho2021multi}, HAPs are increasingly regarded as a foundational technology for future wireless communication systems. Operating in the stratosphere at altitudes of 20–30\,km, HAPs uniquely combine the {broad-area coverage of satellites} with the {maneuverability and flexible deployment} of unmanned aerial vehicles (UAVs), making them especially suitable for long-endurance missions such as persistent connectivity, environmental monitoring, and real-time observation.

Driven by the growing global demand for ubiquitous, high-throughput, and low-latency wireless services, HAPs are poised to play a pivotal role in the development of {integrated satellite–airborne–terrestrial communication networks}.
Compared with conventional satellite constellations, HAPs offer several key advantages, including {lower deployment costs}, {shorter deployment timelines}, and {significantly reduced communication latency}~\cite{Shibata2019AStudyon}. These benefits have placed HAPs at the forefront of both academic research~\cite{dOliveira2016highaltitude} and industrial innovation~\cite{airbuszephyr,JirousekPeichl2023DesignofaSynthetic,JirousekPeichl2024DLR}, reinforcing their strategic relevance in the evolution of next-generation wireless ecosystems.

Although the concept of HAPs has been around for nearly three decades~\cite{wang2014high,aragon2008high,hall1983preliminary}, its large-scale deployment has long been constrained by technological limitations. Earlier efforts were hindered by low solar energy conversion efficiency, the absence of lightweight structural materials, and inadequate avionics and control capabilities. However, recent advances in enabling technologies—such as high-performance composite materials, high-efficiency solar energy systems, antenna beamforming, and autonomous avionics—have collectively revived momentum in HAPs development~\cite{dOliveira2016highaltitude}. Depending on mission-specific requirements, such as payload, power demand, and coverage objectives, modern HAPs platforms can be realized in the form of {balloons, airships}, or {fixed-wing solar-powered aircraft}~\cite{abbasi2024HAPs}.

Recent works have proposed architectural visions and communication frameworks for future HAPs-based systems~\cite{kurtKarabulut2021Avision}, while surveys and technology trend analyses have highlighted the maturity of airframe, power, and payload subsystems required for sustained high-altitude operations~\cite{belmekki2024cellular}. Large-scale stochastic-geometry analyses further show that directional HAPs beams
can deliver reliable wide-area coverage even under dense spatial reuse\,\cite{Lou2025Coverage}.
Complementary cellular enhancements, such as two-tier terrestrial/HAPs
architectures\,\cite{Arum2024TwoTier} and non-orthogonal multiple access (NOMA)-assisted remote-coverage
schemes\,\cite{Javed2025SystemDesign}, underscore the role of HAPs in extending service to
underserved regions.

Beyond radio access, the stratosphere is evolving into a compute continuum.
Hierarchical aerial-edge frameworks now coordinate task off-loading, user
association, and resource orchestration across cooperating UAVs and HAPs
\cite{Nabi2025JointOffloading,Wu2025MultiHAP,Yu2024MECNOMA},
while multi-HAPs clusters embedded in space–air–ground–sea (SAGS) networks
reduce latency for maritime and polar users\,\cite{Wu2025MultiHAP}.
Physical-layer innovations reinforce this evolution.
Laser relays that exploit Fresnel–Poincaré–Kummer optics for terrestrial-HAPs-satellite
links\,\cite{Lu2024EFPK},
adaptive beam-switching protocols for high-speed rail\,\cite{Li2025BeamSwitching},
hemispherical antenna arrays that equalize capacity across footprints\,\cite{Abbasi2024HemiArray},
reconfigurable intelligent surfaces (RIS)-assisted terahertz (THz)  NTNs links\,\cite{Amodu2024RIS}, and multi-HAPs THz relays
robust to I/Q imbalance\,\cite{Ahrazoglu2024MultiHAPs} demonstrate the feasibility of
ultra-high-capacity stratospheric links, while optical inter-HAPs backbones
boost reliability for integrated space–air–ground networks\,\cite{Niu2025Optical}.

Sustainability threads permeate these advances.
Energy-aware radio-access planning\,\cite{Salamat2023Energy},
data-centre-enabled HAPs concepts\,\cite{Abderrahim2024DataCenter},
and cloud-enabled HAPs delivering equitable 6G access\,\cite{Alghamdi2022CloudHAPs}
illustrate the shift toward carbon-neutral, service-oriented NTNs.
Power-allocation strategies for HAPs-assisted LEO backhauls\,\cite{Ali2024PowerLEO},
efficient data-transmission protocols for massive HAPs constellations\,\cite{Wang2024EfficientDT},
and joint user-association/beamforming designs in satellite–HAPs–ground
networks\,\cite{Jia2022UserAssociation} complete the picture of a
maturing ecosystem in which HAPs evolve from simple airborne relays to fully
integrated compute-and-sensing hubs orchestrating resources across the
SAGS continuum.



HAPs combine cost-effective deployment with wide-area coverage, agile repositioning, and multi-month endurance, making them attractive not only for bridging the digital divide but also for mission-critical services such as disaster relief, environmental monitoring, and smart agriculture \cite{abbasi2024HAPs,Lou2024NTN}.
Among the sensing modalities that can unlock these services, synthetic-aperture radar (SAR) is particularly compelling because it delivers high-resolution images day or night, in virtually any weather.  After earthquakes, landslides, or floods, events that often coincide with heavy cloud cover or rain, SAR can rapidly map damage, track terrain deformation, and guide emergency response.  Equipping HAPs with SAR sensors therefore provides a practical path to real-time, resilient situation awareness at continental scales.  Recent work extends this concept even further. HAPs-assisted integrated sensing-and-communication (ISAC) architectures couple high-resolution surveillance with secure data delivery \cite{Wang2025AerialISAC}, while mission-oriented schedulers that coordinate satellites, HAPs, and ground nodes enable specialised applications ranging from oil-spill monitoring \cite{Zhang2025OilSpill} to uninterrupted connectivity along high-speed-rail corridors \cite{Li2025HSRScheduling}.

Motivated by the need for an integrated design that simultaneously exploits the wide-area connectivity of HAPs and the all-weather, high-resolution sensing of SAR, this paper proposes a joint trajectory-and-resource optimization framework for HAPs-SAR missions. Moreover, the contributions of this paper can be summarized as follows:
\begin{itemize}
    \item While most existing HAPs research has focused primarily on communication functions, this paper explicitly incorporates SAR imaging into the system model, an essential capability for applications such as disaster response, environmental monitoring, and smart agriculture. Recognizing the unique circular flight dynamics of HAPs platforms, we propose a novel dual-mode SAR sensing framework that integrates both small-circle circular SAR (CSAR) and large-circle circular trace scanning SAR (CTSSAR) modes. This hybrid configuration leverages the maneuvering flexibility of HAPs to enhance spatial coverage and imaging diversity. To the best of the authors’ knowledge, this is the first work to introduce such a coordinated SAR sensing architecture tailored to the operational characteristics of stratospheric platforms.
    \item Building upon the proposed sensing architecture, this paper presents a mathematical formulation of the HAPs trajectory, alongside an analytical characterization of the SAR sensing region and the necessary conditions for successful image acquisition. Given the stringent energy constraints inherent to stratospheric platforms, real-time onboard SAR processing is typically infeasible, necessitating the offloading of raw echo data to a ground base station (BS) for reconstruction. Accordingly, we also model the communication link required to support reliable and timely data transmission.
From an energy perspective, the HAPs must judiciously allocate its limited solar-harvested power across three competing subsystems, propulsion, energy-intensive SAR payloads, and long-range data downlinks. To holistically capture these interdependent factors, we formulate a unified mixed-integer nonlinear programming (MINLP) framework that jointly optimizes flight trajectory, sensing performance, communication quality, and energy sustainability.
\item  To reconcile these constraints, we develop a successive convex approximation (SCA) algorithm augmented with a finite search over SAR sweeps. This hybrid solver yields iterative, tractable sub-problems whose solutions provably converge to a high-quality operating point.
\item
Comprehensive simulations confirm the effectiveness of the proposed approach. The HAPs adaptively adjusts altitude and heading to maximize radar coverage, which improves energy utilization, and sustains reliable data transmission for real-time, off-board image reconstruction. A final simulation over a 9-target lattice demonstrates accurate SAR imaging and validates the practical feasibility of the joint design.
\end{itemize}

The remainder of this paper is structured as follows. Section \ref{sec2} introduces the system model, including the HAPs trajectory, SAR sensing geometry, communication link, and energy supply, all of which are mathematically characterized. In Section \ref{sec3}, the joint optimization problem is formulated as a MINLP, and a solution approach is developed based on the SCA method. Section \ref{sec4} presents simulation results that validate the proposed framework. Finally, Section \ref{sec5} concludes the paper and discusses future research directions.

\section{System model}
\label{sec2}
Fig.\ref{fig1} illustrates the schematic of the proposed SAR imaging framework for ground observation using a HAPs. The HAPs performs a circular flight around the vertical ($z$) axis and operates in two distinct SAR modes to balance resolution and coverage.
As shown in Fig.\ref{fig1}(a), the platform initially flies along a small-radius circular trajectory and performs CSAR imaging \cite{LinHong2011ExtensionofRange}, which enables high-resolution, full-azimuth imaging of the target scene. To expand the sensing range, the platform gradually increases its flight radius, transitioning into the CTSSAR mode \cite{LiaoWang2016TwoDimensionalSpectrum}, as illustrated in Fig.\ref{fig1}(b). This mode enables scene-wide mapping by scanning concentric circular rings.

Throughout the flight, SAR echo data is transmitted to nearby ground BS for real-time image reconstruction. This design avoids performing onboard SAR processing, which is generally infeasible due to the HAPs's limited energy and computational resources.
To ensure high-quality image acquisition, the CSAR mode requires that the radar beam fully illuminates the entire circular track, enabling full-perspective imaging. In contrast, the CTSSAR mode demands that each circular ring is fully covered to guarantee complete scene mapping. Additionally, to ensure seamless coverage over extended ground areas, it is essential that adjacent circular sweeps are contiguous, without any gaps between beam footprints.
 
\begin{figure}[htp]
	\centering
	\subfigure[]{
		{\includegraphics[width=0.35\textwidth]{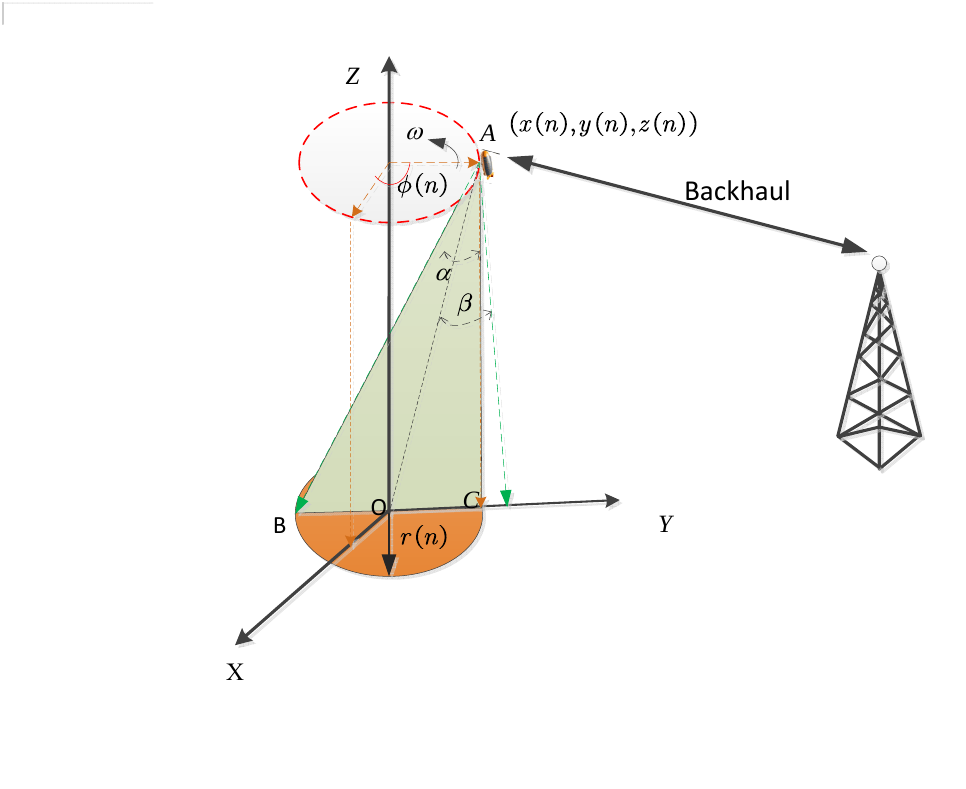}}}
	\subfigure[]{
		{\includegraphics[width=0.31\textwidth]{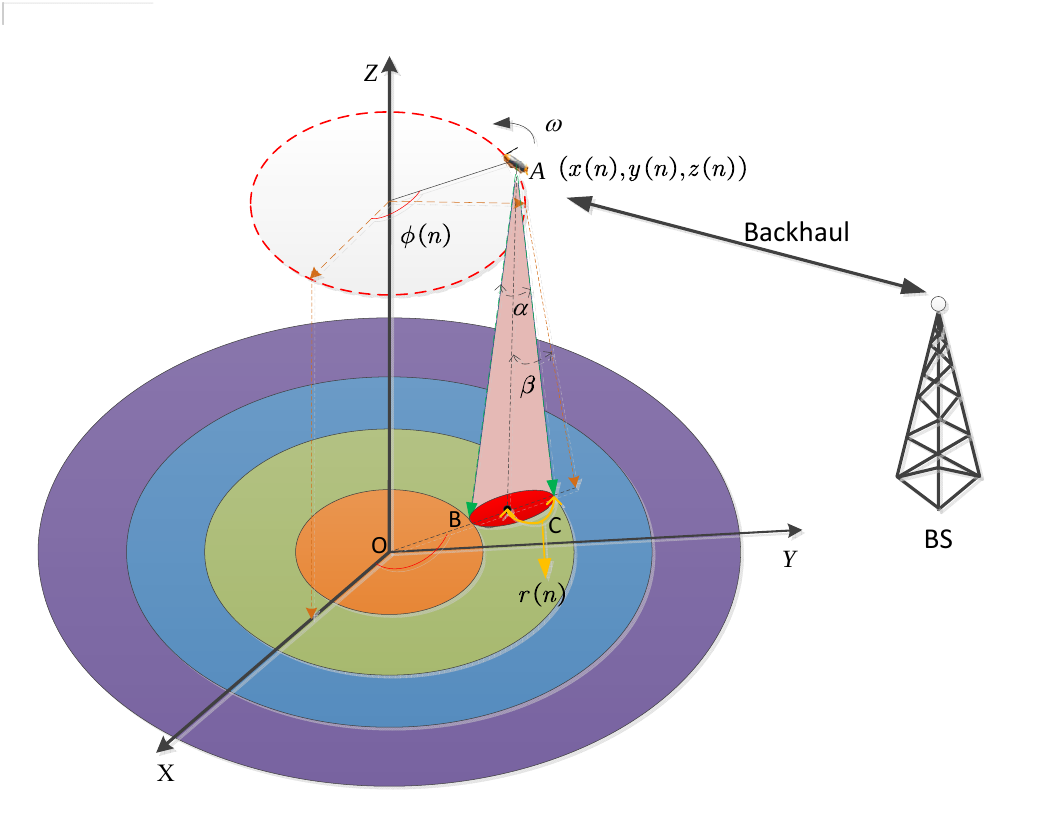}}}
	\caption{Imaging scenario for HAPs-SAR with a backhaul communication link. (a) Circular SAR mode. (b) Circular Trace Scanning SAR mode.}
	\label{fig1}
\end{figure}

\subsection{Design Trajectory}
Consider a HAPs operating above the ground and moving along a circular trajectory of radius \( R \) in a plane parallel to the $xy$ plane. The platform rotates with a constant angular velocity denoted by \( \omega \). At the initial slow-time instant \( \eta = 0 \), the HAPs is assumed to be located on the positive side of the $x$-axis. Accordingly, the azimuth angle relative to the $x$-axis is given by  $\varphi(\eta) = \omega \eta.$
Let \( T \) represent the time required for the HAPs to complete one full circular revolution. This period is uniformly discretized into \( N \) time intervals, where the time step for each segment is defined as  $\Delta = \frac{T}{N}$.
To ensure consistent discretization, \( \Delta \) should be selected such that \( N = T / \Delta \) is an integer.

To complete the prescribed ground mapping strip, the HAPs is assumed to perform \( M \) full circular rotations. To more precisely characterize the temporal evolution of the HAPs trajectory over these circular sweeps, we define the following set of time slots:  
\begin{equation}
	\label{key}
	\mathbb{A} =\left\{ 1,N+1,\ldots,\left( M-1 \right) N+1 \right\}, 
\end{equation}
and
\begin{equation}
	\label{key}
	\mathbb{B} =\left\{ N,2N,\cdots ,MN \right\}, 
\end{equation}
to represent the indices corresponding to the start and end time slots of each circular motion, respectively.  Subsequently, define the set $
\mathbb{C} 
$ as
\begin{equation}
	\label{key}
	\begin{split}
		\mathbb{C} =&\left\{ 1,2,\cdots ,NM \right\} \backslash \left\{ \mathbb{A} \cup\mathbb{B} \right\}
		\\
		=&\left\{ n\in \left\{ 1,2,\cdots ,NM \right\} \left| \left( n-1 \right) mod\,\,N\ne 0,n\,\,mod\,\,N\ne 0 \right. \right\} 
	\end{split}
\end{equation}
which corresponds to the indices of all remaining time intervals. 
The position of the HAPs at time slot \( n \) is represented as $
\left( x\left( n \right) ,y\left( n \right) ,z\left( n \right) \right) 
$ in the three-dimensional (3D) Cartesian coordinate system. Here, \( x\left( n \right) \) and \( y\left( n \right) \) represent the coordinates in the \( xy \)-plane, while \( z\left( n \right) \) denotes the height of the HAPs above the ground. 

The HAPs performs a circular motion around the $z$-axis, returning to the starting point along the $x$-axis. In other words, to ensure the completeness of the ground mapping strip, the HAPs must align the start and end points during each circular motion. 
It is important to note that the time indices in set \( \mathbb{A} \) typically correspond to the time slots in which the distance between the HAPs platform and the $z$-axis changes.

Further, the angle traversed by the HAPs at time slots \( n \) can be represented as 
\begin{align}
	\label{eq4}
	\phi (n)=&2\pi \frac{n-1}{N},n\in \mathbb{A} ,
	\\
	\label{eq5}\phi (n)=&2\pi \frac{n}{N},n\in \mathbb{B} ,
	\\
	\label{eq6}\phi (n)=&\omega n\Delta ,n\in \mathbb{C} ,
\end{align}
where $n=1,2,\ldots,MN$.
\eqref{eq4}-\eqref{eq5} ensure that the start and end points of each circular motion coincide, while \eqref{eq6}
guarantees that the HAPs performs circular motion.

Assume that the center of the SAR signal beam emitted by the HAPs forms an angle \( \beta \) with the vertical axis of the HAPs projected onto the $xy$ plane. The width of the signal beam emitted by the HAPs is \( \alpha \). In this case, it is easy to determine that the distance from the center of the beam to the edge of the beam on the ground can be expressed as
\begin{equation}
	\label{eq8}
	r\left( n \right)=\frac{1}{2}z\left( n \right) \left( \tan \alpha _1-\tan \alpha _2 \right),
\end{equation}
where the angles $\alpha _1$ and $\alpha _2$ are given by $	\alpha _1=\beta +\frac{\alpha}{2}$ and
	$\alpha _2=\beta -\frac{\alpha}{2}$, respectively.

Furthermore, the 3D position coordinates of the HAPs at time slot \( n \) can be represented as

\begin{align}
    \label{key}
&C1:x\left( 1 \right) =z\left( 1 \right) \tan \beta ,
\\
&C2:
x\left( n \right) =x\left( n-1 \right) +2r\left( n \right) +
\dot{z}\left( n \right) ,\forall n\in \mathbb{A} \backslash \left\{ 1 \right\},
\\
&C3:x\left( n \right) =x\left( \left( \lceil \frac{n}{N} \rceil -1 \right) N +1 \right) \cos \left( \varphi \left( n \right) \right) ,\forall n\in \mathbb{B} \cup \mathbb{C} ,
\\
&C4:y\left( n \right) =x\left( 
\left( \lceil \frac{n}{N} \rceil -1 \right) N+1 \right) \sin \left( \varphi \left( n \right) \right) ,
\\
&C5:z\left( n \right) =z\left( n-1 \right) ,\forall n\in \mathbb{B} \cup \mathbb{C} .
\end{align}
where $
\lceil \cdot \rceil 
$ denotes the ceiling function and $
\dot{z}\left( n \right) =\left[ z\left( n \right) -z\left( n-1 \right) \right] \tan \alpha _2
$.  
The constraint $C2$ ensures there are no gaps between the beam coverage areas of adjacent sweeps. Meanwhile, $C3$ and $C4$ ensure that the HAPs follows a circular trajectory.

%
%
%
%
\subsection{SAR Analysis}
This subsection will analyze the ground coverage area for SAR imaging using HAPs. 
\begin{figure}[htp]
	\centering
	\subfigure[]{
		{\includegraphics[width=0.24\textwidth]{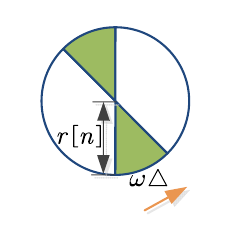}}}
	\subfigure[]{
		{\includegraphics[width=0.22\textwidth]{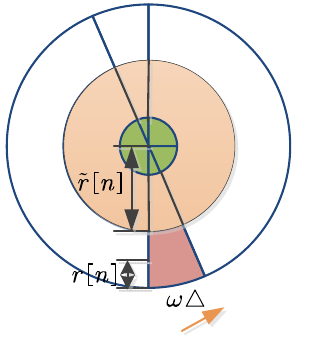}}}
	\caption{Covering area for HAPs. (a) CSAR mode. (b) CTSSAR mode.}
	\label{fig2}
\end{figure}
Furthermore, when HAPs operates in CSAR mode, specifically mapping the smallest circle as shown in Fig. \ref{fig2}(a), it is straightforward to derive that the area covered by the SAR beam within each time interval \(\Delta\) can be expressed as

\begin{equation}
    \label{key}
S_1=\omega \bigtriangleup _{}r^2\left( n \right). 
\end{equation}

Fig. \ref{fig2}(b) shows the imaging results of the HAPs operating in CTSSAR mode for the ground surface. The area scanned by SAR within the time interval \(\Delta\) can be calculated as
\begin{equation}
    \label{key}
    \begin{split}
S_2=&\frac{1}{2}\left[ 2r\left( n \right) +\tilde{r}\left( n \right) \right] ^2\omega \bigtriangleup -\frac{1}{2}\tilde{r}^2\left( n \right) \omega \bigtriangleup 
\\
=&2\left\{ r^2\left( n \right) +r\left( n \right) \tilde{r}\left( n \right) \right\} \omega \bigtriangleup 
    \end{split}
\end{equation}
with expression of $\tilde{r}\left( n \right)$ is given by
\begin{equation}
    \label{key}
\tilde{r}\left( n \right) =\begin{cases}
	r\left( n-N \right) ,\lceil \frac{n}{N} \rceil =2,\\
	r\left( n-N \right) +\sum_{m=2}^{\lceil \frac{n}{N} \rceil -1}{2r\left( n-mN \right) ,\lceil \frac{n}{N} \rceil >2}.\\
\end{cases}
\end{equation}

Hence, the area observed by HAPs during the entire surveying period can be represented as
\begin{equation}
    \label{eq19}
S=\sum_{n=2}^N{S_1}+\sum_{n=N+1,n\notin \mathbb{A}}^{MN}{S_2.}
\end{equation}

The preceding discussion has analyzed the entire imaging coverage area of the HAPs. Next, we will introduce two parameters in detail, the SAR data rate and the minimum signal-to-noise ratio (SNR) required to achieve effective SAR imaging. In this respect, based on the results from \cite{younis2004digital,currie1992wide}, the data rate generated by SAR imaging in time slot \( n \) can be expressed as
\begin{equation}
    \label{key}
D_{\min}\left( n \right) =B_w\left[ \frac{2}{c}\left( R_{\mathrm{far}}\left( n \right) -R_{\mathrm{near}}\left( n \right)  \right) +T_p \right] \cdot PRF\,\,\left[ \mathrm{bit}/\mathrm{s} \right], 
\end{equation}
where $B_w$ represents the bandwidth of the SAR transmitted pulse, and $T_p$ denotes the pulse width. The symbol $c$ stands for the speed of light, and $PRF$ refers to the pulse repetition frequency. $R_{\mathrm{far}}\left( n \right) $ and $R_{\mathrm{near}}\left( n \right) $ indicate the radar's maximum unambiguous range and minimum unambiguous range, respectively, which are given by
\begin{equation}
\label{key}
R_{\mathrm{far}}\left( n \right) =\frac{z\left( n \right)}{\cos \alpha _1},R_{\mathrm{near}}\left( n \right) =\frac{z\left( n \right)}{\cos \alpha _2}.
\end{equation}

Furthermore, based on the radar equation \cite{richards2005fundamentals}, the SNR received by the HAPs receiver is given by
\begin{equation}
    \label{eq19}
\mathrm{SNR}\left( n \right) =\frac{P_{\text{Rad}}(n) G_tG_r\lambda ^3\sigma _0cT_p\mathrm{PRF}}{256\pi ^3R_{0}^{3}KTF_nB_wL_sV(n)\sin \delta}
\end{equation}
where \( P_{\text{Rad}}(n) \) denotes the radar transmission power at the \( n \)-th time slot. \( G_t \) and \( G_r \) represent the transmit and receive gains, respectively. \( \lambda \) is the wavelength, while \( \sigma_0 \) signifies the backscatter coefficient. \( R_{0} \) refers to the distance between the radar and the center of the scene, commonly known as the reference distance. \( K \) stands for Boltzmann's constant, and \( T \) represents the receiver temperature. \( F_n \) denotes the receiver noise, \( L_s \) indicates system loss, and \( V(n) \) is platform velocity, whose expression is\footnote{The velocity expression in \eqref{eq22} is derived solely based on the equilibrium between centrifugal and gravitational forces. This formulation remains valid under the assumption that complex factors such as wind disturbances, lift variations, and other aerodynamic effects are neglected \cite{beard2012small,austin2011unmanned}.} 
\begin{align}
    \label{eq22}
&C6:
V\left( n \right) =\sqrt{x\left( n \right) g\tan \zeta},\forall n\in \mathbb{A}  
\\
&C7:V\left( n \right) =V\left( n-1 \right)
,\forall n\in \mathbb{B} \cup \mathbb{C}. 
\end{align}

Here, the definition of the symbol \(g\) is provided in Table \ref{tab1} and $
\zeta 
$ denotes the banking angle.
Moreover, \( \delta \) represents the beam incidence angle, the angle between the antenna line-of-sight vector and its intersection with the $xy$ plane. In addition, the following relationship can be easily derived from Fig. \ref{fig3}, namely 
\begin{align}
    \label{eq20}
    \delta +\beta =\frac{\pi}{2},
    z\left( n \right) =R_{0}^{}\sin \delta .
\end{align}

\begin{figure}[htp]
	\centering
	\subfigure{
		{\includegraphics[width=0.40\textwidth]{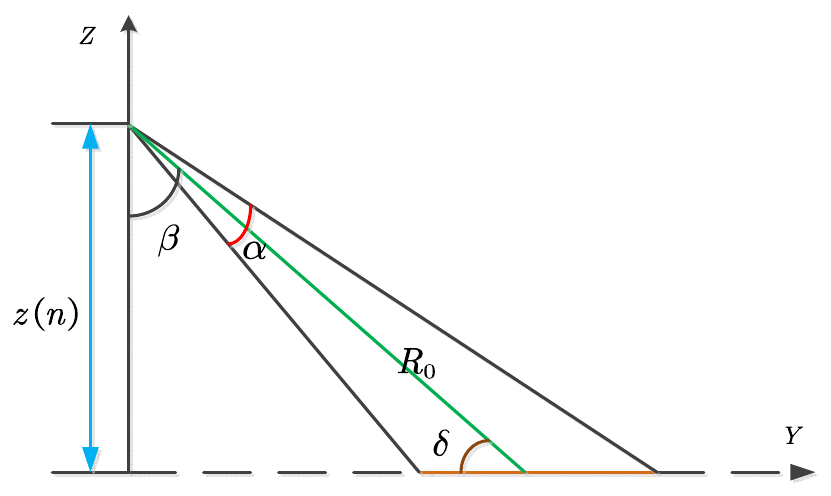}}}
	\caption{Geometric schematic of the HAPs within the $yz$ plane.}
	\label{fig3}
\end{figure}

Moreover, inserting \eqref{eq20} into \eqref{eq19} yields 
\begin{equation}
  \label{equ25}
\mathrm{SNR}\left( n \right) =\frac{P_{Rad}\left[ n \right] G_tG_r\lambda ^3\sigma _0cT_p\mathrm{PRF}\cos ^2\beta}{256\pi ^3z^3\left( n \right) KTF_nB_wL_sV\left( n \right)}.
\end{equation}
\subsection{Communication Link}
HAPs platforms are expected to primarily rely on renewable energy sources, such as solar or wind power, for their long-term energy supply~\cite{belmekki2024cellular}. Given the limited onboard power and processing capabilities, {real-time SAR image reconstruction directly on the HAPs is generally infeasible}. Instead, the raw ground-scattered SAR echoes must be transmitted to a ground BS for post-processing. This requirement establishes a {data backhaul link} between the HAPs and the BS.

To conserve onboard storage and support {real-time ground-based imaging}, it is crucial that the SAR echo data should be transmitted promptly to the BS, rather than accumulated and sent only after the entire mapping operation is complete. This subsection analyzes the characteristics of this air-to-ground communication link.

For modeling purposes, we assume that the HAPs-to-BS backhaul link consists solely of a {line-of-sight (LOS)} channel, and the air-to-ground propagation adheres to a {free-space path loss} model. The channel power gain at time slot \( n \) is then expressed as
\begin{equation}
\label{key}
\rho \left( n \right) =\frac{\rho _0}{d^2\left( n \right)}.
\end{equation}

Here, \( \rho_0 \) denotes the channel gain at a reference distance of 1m, while \( d(n) \) represents the LOS distance between the HAPs and the BS at time slot \( n \). Assuming the BS is positioned at \( (x_s, y_s, z_s) \) in a 3D Cartesian coordinate system, the expression for \( d(n) \) can be written as
\begin{equation}
  \label{key}
d\left( n \right) =\sqrt{\left( x\left( n \right) -x_s \right) ^2+\left( y\left( n \right) -y_s \right) ^2+\left( z\left( n \right) -z_s \right) ^2}.
\end{equation}

If the SAR's transmitted communication power at time slot \( n \) is denoted as \( P_{\mathrm{Com}}(n) \), the received SNR at the BS can be expressed as
\begin{equation}
  \label{key}
\gamma\left( n \right) =\frac{P_{\mathrm{Com}}\left( n \right) \rho \left( n \right)}{\sigma _{0}^{2}}
\end{equation}
with $\sigma _{0}^{2}$ denoting the receiver's noise power. Furthermore, based on the above analysis, the instantaneous backhaul throughput from the HAPs to the BS can be readily given by
\begin{equation}
  \label{key}
D_c\left( n \right) =B_c\log \left( 1+\gamma \left( n \right) \right) 
\end{equation}
where $B_c$ represents the bandwidth of communicating signal. To enable the real-time transmission of SAR sensing data, the following condition must be satisfied
\begin{equation}
    \label{key}
C8: D_c\left( n \right) >D_{\min}\left( n \right). 
\end{equation}

The reason for using $ > $ instead of  $\geqslant $ in $C8$ is that this paper does not account for essential signals, such as clock synchronization and the positional information of the HAPs, which are transmitted alongside the SAR raw echo.
\subsection{Energy Analysis}
The proposed HAPs concept involves operating at higher stratospheric altitudes during daylight hours, serving as an airborne BS for ground-based communication while simultaneously performing sensing tasks by receiving echo signals from the surface. During flight, the HAPs is equipped with solar panels to continuously harvest solar energy, which powers its real-time operations, such as propulsion and data transmission, and stores surplus energy to support nighttime activity.

In typical communication-only configurations, a HAPs may remain quasi-static during nighttime to minimize energy consumption. However, when sensing is included, the platform must remain fully operational 24/7. This requirement is particularly important in emergency scenarios, such as earthquakes, where the HAPs must perform real-time SAR imaging over disaster-stricken areas and immediately transmit sensing data to the ground BS. Such real-time capabilities are essential for enabling timely and accurate decision-making by rescue teams.




In this respect, based on the findings of \cite{javed2023interdisciplinary}, the propulsion power necessary to ensure that the HAPs maintains stable circular motion at time slot $n$ can be expressed as 
\begin{equation}
  \label{eq300}
\begin{split}
P_{\mathrm{Mot}}\left( n \right) =&\frac{1}{\cos ^2\left( \zeta \right) \eta _p\eta _e}\left[ \frac{1}{2}\rho \left( z\left( n \right) \right) V^3\left( n \right) 
\tilde{S}
C \right. 
\\
\,\,          &\qquad\qquad\qquad\qquad
\left. +\varepsilon \frac{2W^2}{\rho \left( z\left( n \right) \right) \tilde{S}V\left( n \right)} \right]
\end{split}
\end{equation}
where  $\eta _p, \eta _e$ represent the efficiencies of propeller and engine, respectively. Meanwhile, $\tilde{S}$ and $W$ are the total area of the wing and the weight of HAPs, respectively. The Zero-lift drag is defined by $C$ and the expression of coefficient is $
\varepsilon =1/\left( \pi eR_w \right) 
$ with $e$ and $R_w$ being the Oswald’s efficiency factor \cite{oswald1993efficient} and wing aspect ratio \cite{jones1940unsteady}, respectively.
 Besides, $\rho \left( z\left( n \right) \right)$ is the air density at an altitude $z\left( n \right)$, given by 
\begin{equation}
  \label{eq29}
\rho \left( z\left( n \right) \right)=\frac{p_h\left( z\left( n \right) \right)}{R_sT_p\left( z\left( n \right) \right)}
\end{equation}
with the expressions of $p_h\left( z\left( n \right) \right)$ and $T_p(z\left( n \right))$ \cite{atmospheric1976us} being
\begin{equation}
  \label{eq30}
p_h\left( z\left( n \right) \right) =p_{b1}\left( \frac{T_b}{T_b+L_b\left( z\left( n \right) -H_1 \right)} \right) ^{\frac{g\Bar{M}}{RL_b}},H_1\leqslant z\left( n \right) \leqslant H_2
\end{equation}
and 
\begin{equation}
    \label{key}
T_p(z\left( n \right) )=T_b+L_b\left( z\left( n \right) -H_1 \right) ,H_1\leqslant z\left( n \right) \leqslant H_2
\end{equation}

Here, all used symbols in  \eqref{eq29} and \eqref{eq30} are summarized in Table \ref{tab1}.

\begin{table}[htp]
	\centering
	\caption{Symbols  used in \eqref{eq29} and \eqref{eq30} }
	\label{tab1}
	\begin{tabular}{c|ccc}
		\hline
		{Symbol} &{Explanation} & Value & Unit
		\\	\cline{1-4}	
		$p_{b1}$	& base static pressures &5474.889 &	Pa
		\\
        $g$ & gravitational acceleration &9.8 &$
\mathrm{m}/\mathrm{s}^2
$
        \\
        $\Bar{M}$ &molar mass of Earth’s air &0.0289644 & $
\mathrm{kg}/\mathrm{mol}
$\\
        $T_b$ &base temperature &216.65&K
        \\
        $L_b$ &base temperature lapse rate &6.7 &K/km\\
        $R$ &universal gas constant &8.31432 &$
\mathrm{N}\cdot \mathrm{m}/\mathrm{mol}\cdot \mathrm{K}
$\\
        $R_s$ &specific gas constant &287.052 &$
\mathrm{J}/\mathrm{kg}\cdot \mathrm{K}
$\\
        $H_1$ &height (lower limit) &20&km\\
        $H_2$ &height (upper limit)&32&km\\
		\hline
	\end{tabular}
	\label{tab1}
\end{table}

Therefore, we can model the combined energy requirements for maintaining the HAPs platform's operations across its propulsion system, radar sensing capabilities, and communication tasks in the given time slot 
$n$, as
\begin{equation}
    \label{key}
P_{\mathrm{Total}}(n)=P_{\mathrm{Mot}}(n)+P_{\mathrm{Rad}}(n)+P_{\mathrm{Com}}(n).
\end{equation}

Next, we proceed to analyze the solar power $
P_{\mathrm{Har}}(n)
$ harvested by the solar panels of HAPs. Drawing upon the findings presented in \cite{reda2004solar,aglietti2009harnessing}, the expression for \( P_{\mathrm{Har}}(n) \) can be elegantly derived as\footnote{Nighttime operation implies zero energy harvesting, i.e., \(P_{\mathrm{Har}}(n) = 0\), which is not considered in this paper.}
\begin{equation}
  \label{eq355}
P_{\mathrm{Har}}(n)=\eta _hAI(z\left( n \right) )
\end{equation}
where the symbol $\eta _h$ denotes the energy conversion efficiency of the solar panel and \( A \) is the effective solar energy collection area of the solar panel. \( I(z(n)) \) represents the solar radiation intensity per unit area, which depends on factors such as the season, date/time, geographical location, and the operational altitude of the HAPs. However, this paper primarily investigates the impact of varying flight altitudes on \( I(z(n)) \), disregarding the influences of observer geographical latitude, solar declination, and local hour angle \cite{reda2004solar} on radiation intensity. Hence, based on the above analysis and the conclusions from \cite{reda2004solar,aglietti2009harnessing,javed2023interdisciplinary}, an approximate expression for \( I(z(n)) \) can be given by
\begin{equation}
    \label{key}
I(z\left( n \right) )=I_0\exp \left\{ -p_{\mathrm{R}}(z\left( n \right) )\alpha _{\mathrm{int}} \right\}. 
\end{equation}
Here, \( I_0 \) represents the standard solar constant at zero air mass, with a value of 1367 W/m². The extinction coefficient, denoted as 
$\alpha _{\mathrm{int}}=0.32$, represents the attenuation of light in clear sky conditions. This value is indicative of the combined effects of scattering and absorption by atmospheric molecules and aerosols under standard clear-sky scenarios. Moreover, The parameter \(p_{\mathrm{R}}(z\left( n \right) )\) is defined as the ratio between the atmospheric pressure at a given altitude \(z\left( n \right)\) and the pressure at ground level. It quantifies the relative pressure at altitude \(z\left( n \right)\) compared to the ground, and is expressed as $p_{\mathrm{R}}(z\left( n \right) )=\frac{p_h\left( z\left( n \right) \right)}{p_0}$.
Here, \(p_h\left( z\left( n \right) \right)\) is the pressure at altitude \(z\left( n \right)\) and
 \(p_0\) is the pressure at ground level, whose value is 101325Pa.

Based on the above analysis, we can easily deduce that the remaining power in time slot \(n\), after accounting for the power $P_{\mathrm{Har}}(n)$ obtained from solar energy, the power required to maintain HAPs aerodynamics $P_{\mathrm{Mot}}(n)$, the power used for radar sensing $P_{\mathrm{Rad}}(n)$, and the power used for communication back to the BS $P_{\mathrm{Com}}(n)$, can be expressed as
\begin{equation}
    \label{key}
    \begin{split}
P_{\mathrm{Re}}(n)=&P_{\mathrm{Har}}(n)-P_{\mathrm{Total}}(n)
\\
=&P_{\mathrm{Har}}(n)-P_{\mathrm{Mot}}(n)-P_{\mathrm{Rad}}(n)-P_{\mathrm{Com}}(n).
    \end{split}
\end{equation}

The variation in the battery energy carried by the HAPs from time slot \(n-1\) to time slot \(n\) can be formulated as
\begin{equation}
    \label{eq39}
C9:E\left( n \right) =E\left( n-1 \right) +\eta _aP_{\mathrm{Re}}(n)\bigtriangleup, 
\forall n\ne 1
\end{equation}
where the parameter $\eta _a$ is defined by 
\begin{equation}
    \label{key}
\eta _a=\left\{ \begin{array}{c}
	\eta _b,P_{\mathrm{Re}}(n)>0\\
	\eta _c,P_{\mathrm{Re}}(n)\leqslant 0\\
\end{array} \right. 
\end{equation}
with \(\eta_b\) and \(\eta_c\) representing the efficiencies of battery charging and discharging, respectively.

Constraint \(C9\) defines the energy allocation framework for the HAPs system, ensuring that energy usage remains within sustainable limits while maximizing overall system performance under available power constraints.


\section{Trajectory and Resource Allocation Optimization}
\label{sec3}
In this section, we first formulate the problem of maximizing the SAR-sensed ground area \( S \), subject to the HAPs trajectory and resource allocation constraints. We then develop a tractable solution approach  based on SCA to address this optimization problem.
\subsection{Objective Formulation}
The objective of this paper is to maximize the SAR ground coverage while satisfying the constraints associated with radar sensing, data transmission, and HAPs mobility. To this end, we jointly optimize the HAPs trajectory parameters, denoted by \( \{ \boldsymbol{z}, N \} \), and the power allocation variables, denoted by \( \{ \boldsymbol{P}_{\mathrm{Mot}}, \boldsymbol{P}_{\mathrm{Rad}}, \boldsymbol{P}_{\mathrm{Com}}, \boldsymbol{E} \} \). Specifically, \( \boldsymbol{z} = \{ z(n) \} \) represents the altitude profile of the HAPs, while \( \boldsymbol{P}_{\mathrm{Mot}}, \boldsymbol{P}_{\mathrm{Rad}}, \boldsymbol{P}_{\mathrm{Com}} \), and \( \boldsymbol{E} = \{ E(n) \} \) denote the time-varying power allocated for mobility, radar sensing, communication, and the remaining onboard energy, respectively, for each time slot \( n \). Based on these definitions, the overall optimization problem can be formulated as 
\begin{equation}
    \label{key}
    \begin{split}
&\left( P1 \right) :\underset{\boldsymbol{z},N,\boldsymbol{P}_{\mathrm{Mot}},\boldsymbol{P}_{\mathrm{Rad}},\boldsymbol{P}_{\mathrm{Com}},\boldsymbol{E}}{\max}\,\,S
\\
&s.t.  \quad C1-C9
\\
&C10: P_{\mathrm{Rad}}(n)=P_{\mathrm{Rad}}(n-1), n\in \left\{ \mathbb{A} \right\} \backslash\left\{ 1 \right\}, 
\\
&C11:0\leqslant P_{\mathrm{Rad}}(n)\leqslant P_{\mathrm{Rad}}^{\max},0\leqslant P_{\mathrm{Com}}(n)\leqslant P_{\mathrm{Com}}^{\max},\forall n,
\\
&C12:\mathrm{SNR}\left( n \right) >\mathrm{SNR}_{\min},\forall n,
\\
&C13:E\left( 1 \right) =E_{\mathrm{Ini}},\\
&C14: E\left( n \right) >0,
\forall n\ne 1
,\\
&C15:z_{\min}\leqslant z\left( n \right) \leqslant z_{\max}, n\in \mathbb{A} ,\\
&C16:
V_{\min}\leqslant V\left( n \right) \leqslant V_{\max},\forall n.
    \end{split}
\end{equation}
Constraints $C1-C5$ define the motion trajectory of the HAPs, 
while $C6$ and $C7$ give  the definition of the speed for the HAPs at time index. 
$C8$ ensures seamless and real-time transmission of the SAR raw data collected by the HAPs to the BS.  The energy limitation of HAPs is encapsulated by constraint $C9$.
Constraint \( C10 \) specifies that the radar's transmitted power remains constant at each moment, ensuring consistency in the radar's operational output. Constraint \( C11 \) indicates that the transmission power of both radar and communication systems must be non-negative and not exceed their respective maximum allowable transmission powers. Constraint \( C12 \) ensures that the echo $\mathrm{SNR}\left( n \right)$ at the HAPs satisfies the minimum SNR threshold required for successful imaging, denoted as \( \mathrm{SNR_{min}} \). Constraints \( C13 \) and \( C14 \) characterize the energy levels of the HAPs at the initial moment  and at time \( n \), respectively. $C15$ and $C16$ represent the operational altitude of HAPs and the speed limit of HAPs, respectively. It is important to note that \( z_{\text{min}} \), \( z_{\text{max}} \), and \( V_{\text{max}} \) are fixed constants, whereas \( V_{\text{min}} \)  denotes the stalling speed, requiring HAPs to sustain a speed no lower than this threshold at the current altitude \( z(n) \) to ensure adequate lift for steady-level flight. Its expression is given by $V_{\min}=\sqrt{\frac{2W}{p_h\left( z\left( n \right) \right) SC_{L}^{\max}}}$,
where $C_{L}^{\max}$ denotes the maximum lift coefficient.

The problem \( \left( P1 \right) \) is a joint non-convex MINLP problem under the given optimization parameters. While the integer constraint on \( N \) does not directly affect the objective function, it influences the dimensionality of other parameters. Consequently, \( N \) cannot be optimized jointly with the other parameters. Therefore, the non-convex problem \( \left( P1 \right) \), involving the integer variable $
N\in \mathbb{N} 
$, is difficult to solve.
 In light of this, this paper proposes a suboptimal solution  based on SCA to address problem \( \left( P1 \right) \).


\subsection{Solution for the Built Problem}
By fixing the variable \( N \) in the first step, we can optimize the problem \( \left( P1 \right) \). In this respect, for a given \( N \), \( \left( P2 \right) \) can be reformulated as 
\begin{equation}
\begin{aligned}
	&\left( P2 \right) :\underset{\boldsymbol{z},\boldsymbol{P}_{\mathrm{Mot}},\boldsymbol{P}_{\mathrm{Rad}},\boldsymbol{P}_{\mathrm{Com}},\boldsymbol{P}_{\mathrm{Har}},\boldsymbol{E}}{\max}\,\,S\\
	&s.t.\quad C1-C16.\\
\end{aligned}
\end{equation}
Even with $ N $  fixed, the presence of constraints $ C8 $, $ C9 $ and $C12 $, continues to render $\left( P2 \right) $ a non-convex problem. To address the non-convexity of problem  $ \left( P2 \right)  $, we adopt the SCA technique to transform it into a tractable sub-optimal convex problem. Furthermore, by proving that the variable $ N $ resides in a finite search space, the optimal value of $ N $ can be efficiently determined via exhaustive search.

We begin by reformulating constraint $C8$ into a convex form. Subsequently, it can be further simplified and equivalently rewritten as
\begin{equation}
    \label{eq41}
    \begin{split}
&C8:\left( \mathfrak{b} 2^{\mathfrak{a} z\left( n \right)}-1 \right) \left[ \left( x\left( n \right) -x_s \right) ^2+\left( y\left( n \right) -y_s \right) ^2 \right. 
\\
&\qquad\qquad\qquad\left. +\left( z\left( n \right) -z_s \right) ^2 \right] <\frac{\rho _0}{\sigma _{0}^{2}}P_{\mathrm{Com}}\left( n \right),\forall n, 
    \end{split}
\end{equation}
with $
\mathfrak{a} =\frac{2B_wPRF}{cB_c}\left( \frac{1}{\cos \alpha _1}-\frac{1}{\cos \alpha _2} \right) 
$ and $
\mathfrak{b} =2^{\frac{B_wT_pPRF}{B_c}}
$ being all positive. 

Denote 
\begin{align}
h_1\left( x\left( n \right) ,z\left( n \right) \right) =&\left( \mathfrak{b} 2^{\mathfrak{a} z\left( n \right)}-1 \right) \left( x\left( n \right) -x_s \right) ^2,
\\
h_2\left( y\left( n \right) ,z\left( n \right) \right) =&\left( \mathfrak{b} 2^{\mathfrak{a} z\left( n \right)}-1 \right) \left( y\left( n \right) -y_s \right) ^2,
\\
h_3\left( z\left( n \right) \right) =&\left( \mathfrak{b} 2^{\mathfrak{a} z\left( n \right)}-1 \right) \left( z\left( n \right) -z_s \right) ^2.
\end{align}

Hence, \eqref{eq41} can be rewritten as 
\begin{equation}
    \label{eq52}
    \begin{split}
        C8:h_1\left( x\left( n \right) ,z\left( n \right) \right) +&h_2\left( y\left( n \right) ,z\left( n \right) \right) +h_3\left( z\left( n \right) \right) \\
        &\qquad\qquad<\frac{\rho _0}{\sigma _{0}^{2}}P_{\mathrm{Com}}\left( n \right) ,\forall n.
    \end{split}
\end{equation}

Through the analysis of the Hessian matrix of the function \( h_1\left( x\left( n \right) ,z\left( n \right) \right) \), it can be concluded that \( h_1\left( x\left( n \right) ,z\left( n \right) \right) \) is not a convex function. However, we find that \( h_1\left( x\left( n \right) ,z\left( n \right) \right) \) can be expressed as the difference of convex functions (DC decomposition). Specifically, we have 
\begin{equation}
    \label{eq53}
    \begin{split}
h_1\left( x\left( n \right) ,z\left( n \right) \right) =&\frac{1}{2}\left[ g_1\left( z\left( n \right) \right) +g_2\left( x\left( n \right) \right) \right] ^2
\\
&-\frac{1}{2}\left[ g_{1}^{2}\left( z\left( n \right) \right) +g_{2}^{2}\left( x\left( n \right) \right) \right] , 
    \end{split}
\end{equation}
with $
g_1\left( z\left( n \right) \right) =\mathfrak{b} 2^{\mathfrak{a} z\left( n \right)}-1,g_2\left( x\left( n \right) \right) =\left( x\left( n \right) -x_s \right) ^2
$. Note that $
g_1\left( z\left( n \right) \right) 
$ and $
g_2\left( x\left( n \right) \right) 
$ are both convex functions. Building upon this, by employing the first-order Taylor expansion, the approximate representations of $
g_{1}^{2}\left( z\left( n \right) \right) 
$, $g_{2}^{2}\left( x\left( n \right) \right)$ around point $
\left\{ x^k\left( n \right) ,y^k\left( n \right), z^k\left( n \right) \right\} 
$
 are derived as
\begin{align}
    \label{key}
g_{1}^{2}\left( z\left( n \right) \right) &\approx g_{1}^{2}\left( z^k\left( n \right) \right) +\left[ g_{1}^{2}\left( z^k\left( n \right) \right) \right] ^{\prime}\left[ z\left( n \right) -z^k\left( n \right) \right], 
\\
g_{2}^{2}\left( x\left( n \right) \right) &\approx g_{2}^{2}\left( x^k\left( n \right) \right) +\left[ g_{2}^{2}\left( x^k\left( n \right) \right) \right] ^{\prime}\left[ x\left( n \right) -x^k\left( n \right) \right]. 
\end{align}
Here, the superscript $ \left( \cdot \right)^{\prime} $ signifies the operation of computing the first derivative of a function.

Similarly, we have 
\begin{align}
    \label{key}
    \begin{split}
      h_2\left( y\left( n \right) ,z\left( n \right) \right) =&\frac{1}{2}\left[ g_1\left( z\left( n \right) \right) +g_3\left( y\left( n \right) \right) \right] ^2
\\
&-\frac{1}{2}\left[ g_{1}^{2}\left( z\left( n \right) \right) +g_{3}^{2}\left( y\left( n \right) \right) \right], 
    \end{split}\\  
\begin{split}
   h_3\left( z\left( n \right) \right) =&\frac{1}{2}\left[ g_1\left( z\left( n \right) \right) +g_4\left( z\left( n \right) \right) \right] ^2
\\
&-\frac{1}{2}\left[ g_{1}^{2}\left( z\left( n \right) \right) +g_{4}^{2}\left( z\left( n \right) \right) \right]  
\end{split}
\end{align}
where $
g_3\left( y\left( n \right) \right) =\left( y\left( n \right) -y_s \right) ^2$ and
$
g_4\left( z\left( n \right) \right) =\left( z\left( n \right) -z_s \right) ^2
$. Moreover, the first-order Taylor expansions of \( g_3\left( y\left( n \right) \right)  \) and \( g_4\left( z\left( n \right) \right)  \) around the point \( 
\left\{ x^k\left( n \right) ,y^k\left( n \right), z^k\left( n \right) \right\} 
 \) can be elegantly given by
\begin{align}
    \label{key}
g_{3}^{2}\left( y\left( n \right) \right) \approx g_{3}^{2}\left( y^k\left( n \right) \right) 
+\left[ g_{3}^{2}\left( y^k\left( n \right) \right) \right] ^{\prime}\left[ y\left( n \right) -y^k\left( n \right) \right], \\
   g_{4}^{2}\left( z\left( n \right) \right) \approx g_{4}^{2}\left( z^k\left( n \right) \right) 
+\left[ g_{4}^{2}\left( z^k\left( n \right) \right) \right] ^{\prime}\left[ z\left( n \right) -z^k\left( n \right) \right]. \label{eq59}
\end{align}

Inserting \eqref{eq53}-\eqref{eq59} into \eqref{eq52} results in
\begin{equation}
    \label{key}
    \begin{split}
&\sum_{i=2}^4{\left[ g_1\left( \varsigma _{4}^{} \right) +g_i\left( \varsigma _i \right) \right] ^2}-3\left[ g_{1}^{2}\left( \varsigma _{4}^{k} \right) +\left[ g_{1}^{2}\left( \varsigma _{4}^{k} \right) \right] ^{\prime}\left( \varsigma _{4}^{}-\varsigma _{4}^{k} \right) \right] 
\\
&-\sum_{i=2}^4{\left\{ g_{i}^{2}\left( \varsigma _{i}^{k} \right) +\left[ g_{i}^{2}\left( \varsigma _{i}^{k} \right) \right] ^{\prime}\left( \varsigma _{i}^{}-\varsigma _{i}^{k} \right) \right\}}<\frac{2\rho _0}{\sigma _{0}^{2}}P_{\mathrm{Com}}\left( n \right) ,\forall n
    \end{split}
\end{equation}
with $
\varsigma _{2}^{}=x\left( n \right) ,\varsigma _{3}^{}=y\left( n \right)$  and $\varsigma _{4}^{}=z\left( n \right) .
$ To further streamline the resolution of the optimization problem $\left( P2 \right)$, this paper introduces three slack vectors, \( 
\boldsymbol{p}=\left[ p\left( 1 \right) ,p\left( 2 \right) ,\cdots ,p\left( MN \right) \right] ^T
\in \mathbb{R} ^{MN}\), \( 
\boldsymbol{q}=\left[ q\left( 1 \right) ,q\left( 2 \right) ,\cdots ,q\left( MN \right) \right] ^T\in \mathbb{R} ^{MN}
 \) and \( 
\boldsymbol{r}=\left[ r\left( 1 \right) ,r\left( 2 \right) ,\cdots ,r\left( MN \right) \right] ^T\in \mathbb{R} ^{MN}
 \), at this juncture. Namely, the constraint $C8$ can be reformulated as
\begin{align}
    \label{key}
    \begin{split}
&\widetilde{C8}:p^2\left( n \right) +q^2\left( n \right) +r^2\left( n \right) -3\left[ g_{1}^{2}\left( \varsigma _{4}^{k} \right) +\left[ g_{1}^{2}\left( \varsigma _{4}^{k} \right) \right] ^{\prime}\left( \varsigma _{4}^{}-\varsigma _{4}^{k} \right) \right] 
\\
&-\sum_{i=2}^4{\left\{ g_{i}^{2}\left( \varsigma _{i}^{k} \right) +\left[ g_{i}^{2}\left( \varsigma _{i}^{k} \right) \right] ^{\prime}\left( \varsigma _{i}^{}-\varsigma _{i}^{k} \right) \right\}}<\frac{2\rho _0}{\sigma _{0}^{2}}P_{\mathrm{Com}}\left( n \right) ,\forall n,
    \end{split}  \\  &\widetilde{C8a}:g_1\left( z\left( n \right) \right) +g_2\left( x\left( n \right) \right) \leqslant p\left( n \right) ,\forall n,
\\
&\widetilde{C8b}:g_1\left( z\left( n \right) \right) +g_3\left( y\left( n \right) \right) \leqslant q\left( n \right) ,\forall n,
\\
&\widetilde{C8c}:g_1\left( z\left( n \right) \right) +g_4\left( z\left( n \right) \right) \leqslant r\left( n \right) ,\forall n.
\end{align}
    
Now, the subproblem \( \left(P2 \right) \) can thus be approximated as the following optimization problem
\begin{equation}
    \label{key}
\begin{aligned}
	&\left( {P}3 \right) :\underset{\boldsymbol{z},\boldsymbol{P}_{\mathrm{Mot}},\boldsymbol{P}_{\mathrm{Rad}},\boldsymbol{P}_{\mathrm{Com}},\boldsymbol{P}_{\mathrm{Har}},\boldsymbol{E}}{\max}\,\,S\\
	&s.t.\quad C1-C7,\widetilde{C8},\widetilde{C8a},\widetilde{C8b},\widetilde{C8c},C9-C16.\\
\end{aligned}
\end{equation}

Note that problem \( (P3) \) remains non-convex due to the presence of constraints \( C9 \) and \( C12 \). In the following, we focus on addressing the non-convexity introduced by constraint \( C9 \). Accordingly, constraint \( C9 \) can be equivalently reformulated as:
 
\begin{equation}
    \label{eq63}
\begin{split}
	\overrightarrow{C9}:&P_{\mathrm{Mot}}(n)-P_{\mathrm{Har}}(n)=\frac{E\left( n-1 \right)}{\bigtriangleup \eta _a}\\
	&-\frac{E\left( n \right)}{\bigtriangleup \eta _a}-P_{\mathrm{Rad}}(n)-P_{\mathrm{Com}}(n),\forall n\ne 1.\\
\end{split}
\end{equation}

At this stage, our attention can be directed toward \( P_{\mathrm{Mot}}(n) \). Accordingly, define 
\begin{align}
    \label{}
h_4\left( z\left( n \right) ,V\left( n \right) \right) =&\frac{1}{2\cos ^2\left( \zeta \right) \eta _p\eta _e}\rho \left( z\left( n \right) \right) V^3\left( n \right) SC,
\\
\label{}
h_5\left( z\left( n \right) ,V\left( n \right) \right) =&\frac{2\varepsilon W^2}{\cos ^2\left( \zeta \right) \eta _p\eta _e\rho \left( z\left( n \right) \right) SV\left( n \right)}.
\end{align}

Furthermore, utilizing the results \cite{javed2023interdisciplinary} clearly reveals that \( h_4(z(n), V(n)) \) and \( h_5(z(n), V(n)) \) are jointly nonconvex functions with respect to $z(n)$ and $V(n)$, respectively. Nevertheless, both \( h_4(z(n), V(n)) \) and \( h_5(z(n), V(n)) \) can be reformulated within the framework of DC programming, that is
\begin{align}
    \label{eq66}
    \begin{split}
h_4\left( z\left( n \right) ,V\left( n \right) \right) =&\frac{1}{2}\left[ g_{5}^{}\left( z\left( n \right) \right) +g_{6}^{}\left( V\left( n \right) \right) \right] ^2
\\
&-\frac{1}{2}\left[ g_{5}^{2}\left( z\left( n \right) \right) +g_{6}^{2}\left( V\left( n \right) \right) \right], 
    \end{split}\\
     \label{eq67}
    \begin{split}
h_5\left( z\left( n \right) ,V\left( n \right) \right) =&\frac{1}{2}\left[ g_7\left( z\left( n \right) \right) +g_{8}^{}\left( V\left( n \right) \right) \right] ^2
\\
&-\frac{1}{2}\left[ g_{7}^{2}\left( z\left( n \right) \right) +g_{8}^{2}\left( V\left( n \right) \right) \right] 
    \end{split}
\end{align}
where $
g_5\left( z\left( n \right) \right) =\rho \left( z\left( n \right) \right)$, $g_6\left( V\left( n \right) \right) =\frac{SCV^3\left( n \right)}{2\cos ^2\left( \zeta \right) \eta _p\eta _e}$, $g_7\left( z\left( n \right) \right) =\frac{1}{\rho \left( z\left( n \right) \right)}$ and $g_8\left( V\left( n \right) \right) =\frac{2\varepsilon W^2}{\cos ^2\left( \zeta \right) \eta _p\eta _eSV\left( n \right)}
$. As demonstrated in Appendix A of \cite{javed2023interdisciplinary}, the functions \( g_5\left( z(n) \right) \), \( g_6\left( V(n) \right) \), \( g_7\left( z(n) \right) \), and \( g_8\left( V(n) \right) \) are all proven to be convex. Furthermore, by applying a first-order Taylor expansion, the functions \( g_5\left( z\left( n \right) \right) \), \( g_6\left( V\left( n \right) \right) \), \( g_7\left( z\left( n \right) \right) \), and \( g_8\left( V\left( n \right) \right) \) can be locally linearized around the point \( \left( z^k\left( n \right), V^k\left( n \right) \right) \), as
\begin{equation}
    \label{eq68}
g_{i}^{2}\left( \upsilon _i \right) \approx g_{i}^{2}\left( \upsilon _{i}^{k} \right) +\left[ g_{i}^{2}\left( \upsilon _{i}^{k} \right) \right] ^{\prime}\left[ \upsilon _{i}^{}-\upsilon _{i}^{k} \right] ,i=5,6,7,8,
\end{equation}
with $
\upsilon _5=\upsilon _7=z\left( n \right) $ and $\upsilon _6=\upsilon _8=V\left( n \right) 
$. 

Subsequently, inserting \eqref{eq66}-\eqref{eq68} into \eqref{eq63} gives 
\begin{equation}
    \label{eq69}
    \begin{split}
&
\widetilde{C9}:\frac{1}{2}\left[ t^2\left( n \right) +u^2\left( n \right) \right] -P_{\mathrm{Har}}(n)+P_{\mathrm{Rad}}(n)
\\
&\qquad-\frac{1}{2}\sum_{i=5}^8{\left[ g_{i}^{2}\left( \upsilon _{i}^{k} \right) +\left[ g_{i}^{2}\left( \upsilon _{i}^{k} \right) \right] ^{\prime}\left[ \upsilon _{i}^{}-\upsilon _{i}^{k} \right] \right] \qquad}
\\
&\qquad\leqslant \frac{E\left( n-1 \right)}{\bigtriangleup \eta _a}-\frac{E\left( n \right)}{\bigtriangleup \eta _a}-P_{\mathrm{Com}}(n),\forall n\ne 1.
    \end{split}
\end{equation}

Here,  \( t(n) \) and \( u(n) \) are components of the slack vectors  
$
\boldsymbol{t} = \left[ t(1), t(2), \cdots, t(MN) \right]^T \in \mathbb{R}^{MN}$ and $
\boldsymbol{u} = \left[ u(1), u(2), \cdots, u(MN) \right]^T \in \mathbb{R}^{MN}.
$
It is important to emphasize that \( \boldsymbol{t} \) and \( \boldsymbol{u} \) serve as slack variables introduced to enhance the tractability of the optimization problem \( (P3) \), and they are subject to 
\begin{align}
    \label{}
\widetilde{C9a}:&g_{5}^{}\left( z\left( n \right) \right) +g_{6}^{}\left( V\left( n \right) \right) \leqslant t\left( n \right) ,
\\
\label{}
\widetilde{C9b}:&g_7\left( z\left( n \right) \right) +g_{8}^{}\left( V\left( n \right) \right) \leqslant u\left( n \right) .
\end{align}

Therefore, the optimization problem $\left( P3 \right) $ can be approximated as 
\begin{equation}
    \label{}
\begin{aligned}
	&\left( P4 \right) :\underset{\boldsymbol{z},\boldsymbol{P}_{\mathrm{Mot}},\boldsymbol{P}_{\mathrm{Rad}},\boldsymbol{P}_{\mathrm{Com}},\boldsymbol{P}_{\mathrm{Har}},\boldsymbol{E}}{\max}\,\,S\\
	&s.t.\quad C1-C7,\widetilde{C8},\widetilde{C8a},\widetilde{C8b},\widetilde{C8c},\\&\qquad\quad\widetilde{C9},\widetilde{C9a},\widetilde{C9b},C10-C16.\\
\end{aligned}
\end{equation}

It is worth noting that constraint \( {C12} \) renders the optimization problem \( (P4) \) non-convex. In this regard, our primary focus is on reformulating constraint $C12$ into a convex form to facilitate efficient optimization. Specifically,
 because the constraint $C12$ involves a cubic polynomial in the HAPs altitude $\boldsymbol{z}$, we can employ second-order cone programming (SOCP) to reformulate $C12$ accordingly. Hence, two slack variables $
\varphi \left( n \right) ,\psi \left( n \right) 
$ are introduced here in vectors $
\boldsymbol{\varphi }=\left[ \varphi \left( 1 \right) ,\cdots ,\varphi \left( MN \right) \right] ^T\in \mathbb{R} ^{MN\times 1}
$ and $
\boldsymbol{\psi }=\left[ \psi \left( 1 \right) ,\cdots ,\psi \left( MN \right) \right] ^T\in \mathbb{R} ^{MN\times 1}
$ to rewrite the constraint C12 as  \cite{boyd2004convex}
\begin{align}
\exists \varphi \left( n \right) \geqslant 0\in \mathbb{R} ,&\widetilde{C12a}: \left[ \begin{matrix}
	\varphi \left( n \right)&		z\left( n \right)\\
	z\left( n \right)&		1\\
\end{matrix} \right] \succeq \mathbf{0},\forall n
\\
\exists \psi \left( n \right) \geqslant 0\in \mathbb{R} ,&\widetilde{C12b}: \left[ \begin{matrix}
	\psi \left( n \right)&		\varphi \left( n \right)\\
	\varphi \left( n \right)&		z\left( n \right)\\
\end{matrix} \right] \succeq \mathbf{0},\forall n
\\
\,\,                   &\widetilde{C12c}:\left[ \begin{matrix}
	P_{\mathrm{Rad}}\left[ n \right] \mathrm{con.}&		V\left( n \right)\\
	\psi \left( n \right)&		1\\
\end{matrix} \right] \succeq \mathbf{0},\forall n
\end{align}
where the symbol $\mathrm{con.}$ denoting the constant, whose expression is $
\mathrm{con.}=\frac{G_tG_r\lambda ^3\sigma _0cT_p\mathrm{PRF}\cos ^2\beta}{256\pi ^3KTF_nB_wL_s\mathrm{SNR}_{\min}}
$.

As a result of the aforementioned transformations, the subproblem \( \left( P1 \right) \) can be equivalently reformulated into a convex optimization problem, denoted as \( \left( P5 \right) \), namely
\begin{equation}
    \label{}
    \begin{split}
      &\left( P5 \right) :\underset{\boldsymbol{z},\boldsymbol{P}_{\mathrm{Mot}},\boldsymbol{P}_{\mathrm{Rad}},\boldsymbol{P}_{\mathrm{Com}},\boldsymbol{P}_{\mathrm{Har}},\boldsymbol{E}}{\max}\,\,S
\\
&s.t.\quad C1-C7,\widetilde{C8},\widetilde{C8a},\widetilde{C8b},\widetilde{C8c},
\\
&\qquad\widetilde{C9},\widetilde{C9a},\widetilde{C9b},C10-C11,
\\
&
\qquad\widetilde{C12a},\widetilde{C12b},\widetilde{C12c}
,C13-C16.   
    \end{split}
\end{equation}
\subsection{Method for optimal Number $
N^*
$}
To tackle Problem $\left( P1 \right)$, we proceed in two steps. First, we solve sub-problem $\left( P5 \right)$ for every admissible fixed value of \(N\). Next, from the resulting solutions we select the smallest \(N\) that still attains the maximum coverage area, thereby minimizing energy consumption. We can now invoke the following proposition to rigorously demonstrate that \( N \) admits only a finite number of feasible values. Consequently, the optimal \( N^* \) can be efficiently determined by performing a finite search.
\begin{proposition}
The sweep $N$ is subject to an upper bound.
\end{proposition}

\begin{proof}

By leveraging constraint $C11$, a lower bound for the right-hand side of constraint $\overrightarrow{C9}$ can readily be established as
\begin{equation}
    \label{}
P_{\mathrm{Mot}}(n)-P_{\mathrm{Har}}(n)\leqslant \frac{E\left( n-1 \right)}{\bigtriangleup \eta _a}-\frac{E\left( n \right)}{\bigtriangleup \eta _a},\forall n\ne 1.
\end{equation}

Subsequently, by summing both sides of the inequality, gives
\begin{equation}
    \label{eq78}
    \sum_{n=2}^{NM}{\left[ P_{\mathrm{Mot}}(n)-P_{\mathrm{Har}}(n) \right]}\leqslant \sum_{n=2}^{NM}{\left[ \frac{E\left( n-1 \right)}{\bigtriangleup \eta _a}-\frac{E\left( n \right)}{\bigtriangleup \eta _a} \right]}.
\end{equation}

Observing \eqref{eq78}, the expression on the right-hand side can be rewritten as 
\begin{equation}
    \label{eq79}
\sum_{n=2}^{NM}{\left[ \frac{E\left( n-1 \right)}{\bigtriangleup \eta _a}-\frac{E\left( n \right)}{\bigtriangleup \eta _a} \right]}=\frac{E\left( 1 \right)}{\bigtriangleup \eta _a}-\frac{E\left( NM \right)}{\bigtriangleup \eta _a}.
\end{equation}

Inserting \eqref{eq79} into \eqref{eq78} results in
\begin{equation}
    \label{}
\sum_{n=2}^{NM}{\left[ P_{\mathrm{Mot}}(n)-P_{\mathrm{Har}}(n) \right]}\leqslant \frac{E\left( 1 \right)}{\bigtriangleup \eta _a}-\frac{E\left( NM \right)}{\bigtriangleup \eta _a}\leqslant \frac{E_{\mathrm{Ini}}}{\bigtriangleup \eta _a}.
\end{equation}

It is worth noting that during each sweep, the altitude and velocity remain constant, leading to fixed values for both power consumption and energy harvesting. Accordingly, we can readily derive\footnote{The validity of \eqref{eq81} is based on the condition that the energy consumed for propulsion exceeds the energy harvested by the solar panels. This condition is practically reasonable, particularly during periods of low solar irradiance, such as early morning or late evening.
}
\begin{equation}
    \label{eq81}
N\leqslant \frac{E_{\mathrm{Ini}}}{\bigtriangleup \eta _a\sum_{m=1}^{M}{\left[ P_{\mathrm{Mot}}(m)-P_{\mathrm{Har}}(m) \right]}}.
\end{equation}

\eqref{eq81} suggests that a feasible upper bound for \(N\) can be determined by solving for the minimum value of $\sum_{m=1}^{M}{\left[ P_{\mathrm{Mot}}(m)-P_{\mathrm{Har}}(m) \right]}$. Hence, we can build the following problem, as
\begin{equation}
    \label{}
    \begin{split}
&\left( P6 \right) :\underset{z\left( m \right) ,V\left( m \right) ,\tilde{t}}{\min}\tilde{t}
\\
&s.t.C15:z_{\min}\leqslant z\left( m \right) \leqslant z_{\max},
\\
&\quad C16:V_{\min}\leqslant V\left( m \right) \leqslant V_{\max},
\\
&\quad C17:\sum_{m=1}^{M}{\left[ P_{\mathrm{Mot}}(m)-P_{\mathrm{Har}}(m) \right]}\leqslant \tilde{t}.
    \end{split}
\end{equation}

The non-convexity of problem $\left( P6 \right)$ arises from the presence of constraint $C17$. To tackle this challenge, we adopt the SCA method, as previously employed. In this manner, problem $\left( P6 \right)$ can be equivalently reformulated into the following convex problem $\left( P7 \right)$
\begin{equation}
    \label{}
    \begin{split}
	&\left( P7 \right) :\underset{z\left( m \right) ,V\left( m \right) ,\tilde{t}}{\min}\tilde{t}\\
	\,\,\\
&
s.t. \widetilde{C9a},\widetilde{C9b},C15,C16,
\\
&
\quad\widetilde{C17}:\sum_{m=1}^{M}{\left[ \frac{1}{2}\left[ t^2\left( m \right) +u^2\left( m \right) \right]  -P_{\mathrm{Har}}(m)\right.}
\\
&
\qquad   - \left.\frac{1}{2}\sum_{i=5}^8{\left[ g_{i}^{2}\left( \upsilon _{i}^{k} \right) +\left[ g_{i}^{2}\left( \upsilon _{i}^{k} \right) \right] ^{\prime}\left[ \upsilon _{i}^{}-\upsilon _{i}^{k} \right] \right]} \right]\leqslant \tilde{t} 
\\
    \end{split}
\end{equation}

An upper bound for the sweep number $N$ can be further determined by solving the optimization problem (P7).
Thus, the proposition is proved.
\end{proof}

To address problem $\left( P1 \right)$, the SCA method is employed to iteratively solve the subproblems $\left( P5 \right)$ and $\left( P7 \right)$, as outlined in Algorithm 1. Both subproblems can be efficiently handled using the conventional solvers, for example the CVX \cite{grant2014cvx}.

\begin{algorithm}
\caption{Overall Optimization Framework Based on SCA}
\label{alg1}
\begin{algorithmic}[1]
\STATE \textbf{Initialization:} $
k=1
$, and the initial point $
\left( \boldsymbol{z}_{\mathrm{}}^{1},\boldsymbol{V}^{1} \right)
\in \mathbb{R} ^{M\times 2}
$, obtain the upper bound $N^{1}$ by solving \( \left( P7 \right) \)  around the point $
\left( \boldsymbol{z}_{\mathrm{}}^{1},\boldsymbol{V}^{1} \right) 
$, $
\xi_1 =\infty 
$, $
\xi_2 =\infty 
$, \textbf{Set} tolerance $
\delta_1 
$, $
\delta_2 
$
\WHILE{$
\xi_1 \geqslant \delta_1 
$}
\STATE Set $k=k+1$
\STATE Get the upper bound $N^{k}$ and solutions $
\left( \boldsymbol{z}_{\mathrm{}}^{k},\boldsymbol{V}^{k} \right) 
$ by solving \( \left( P7 \right) \) around the point $
\left( \boldsymbol{z}_{\mathrm{}}^{k-1},\boldsymbol{V}^{k-1} \right) 
$

\STATE Update $
\xi _1=\left| \frac{N^{ k }-N^{ k-1 }}{N^{k }} \right|
$
\ENDWHILE
\STATE Set $N^{up}=N^k$, 
$
S^*=-\infty $
\FOR{$N=1$ to $N^{up}$}
\STATE Set $k=1$, and the initial point $
\left( \boldsymbol{x}^1,\boldsymbol{y}_{}^{1},\boldsymbol{z}_{}^{1},\boldsymbol{V}^1 \right) \in \mathbb{R} ^{N\times 4}
$, calculate  the  coverage
$S^{1}$ by solving \( \left( P5 \right) \) around the point $
\left( \boldsymbol{x}^1,\boldsymbol{y}_{}^{1},\boldsymbol{z}_{}^{1},\boldsymbol{V}^1 \right) \in \mathbb{R} ^{N\times 4}
$
\REPEAT 
\STATE Set $k=k+1$
\STATE Get the coverage $S^{k}$ and solutions $
\left( \boldsymbol{x},\boldsymbol{y},\boldsymbol{z},\boldsymbol{V},\boldsymbol{P}_{\mathrm{Mot}},\boldsymbol{P}_{\mathrm{Rad}},\boldsymbol{P}_{\mathrm{Com}},\boldsymbol{P}_{\mathrm{Har}},\boldsymbol{E} \right) 
$ by solving \( \left( P5 \right) \) around the point $
\left( \boldsymbol{x}^{k-1},\boldsymbol{y}_{\mathrm{}}^{k-1},\boldsymbol{z}_{\mathrm{}}^{k-1},\boldsymbol{V}^{k-1} \right) 
$
\STATE Set $
\left( \boldsymbol{x}^k,\boldsymbol{y}^k,\boldsymbol{z}^k,\boldsymbol{V}^k \right) =\left( \boldsymbol{x},\boldsymbol{y},\boldsymbol{z},\boldsymbol{V} \right) 
$
\STATE Computate $
\xi _2=\left| \frac{S^{k }-S^{ k-1 }}{S^{k }} \right|
$
\UNTIL{$
\xi _2\leqslant \delta _2
$}
\STATE Set $
\xi_2 =\infty 
$

\IF{$S^{ k}> S^{ *}$}
        \STATE Set $S^{ *} = S^{ k}$
        \STATE Set $
\left( \boldsymbol{x}^*,\boldsymbol{y}^*,\boldsymbol{z}^*,\boldsymbol{V}^*,\boldsymbol{P}_{\mathrm{Mot}}^{*},\boldsymbol{P}_{\mathrm{Rad}}^{*},\boldsymbol{P}_{\mathrm{Com}}^{*},\boldsymbol{P}_{\mathrm{Har}}^{*},\boldsymbol{E}^* \right) =\left( \boldsymbol{x}^k,\boldsymbol{y}^k,\boldsymbol{z}^k,\boldsymbol{V}^k,\boldsymbol{P}_{\mathrm{Mot}}^{k},\boldsymbol{P}_{\mathrm{Rad}}^{k},\boldsymbol{P}_{\mathrm{Com}}^{k},\boldsymbol{P}_{\mathrm{Har}}^{k},\boldsymbol{E}^k \right) 
$

\ENDIF
\ENDFOR
\RETURN $
S^*,\left( \boldsymbol{x}^*,\boldsymbol{y}^*,\boldsymbol{z}^*,\boldsymbol{V}^*,\boldsymbol{P}_{\mathrm{Mot}}^{*},\boldsymbol{P}_{\mathrm{Rad}}^{*},\boldsymbol{P}_{\mathrm{Com}}^{*},\boldsymbol{P}_{\mathrm{Har}}^{*},\boldsymbol{E}^* \right) 
$
\end{algorithmic}
\end{algorithm}
\section{Simulation Results}
\label{sec4}

In this section, we employ numerical simulations to rigorously validate the proposed approach, focusing on three-dimensional HAPs trajectory planning, resource allocation, and algorithm convergence. Unless otherwise specified, the system parameters are listed in Table \ref{tab2}. 
\begin{table}[htp]
	\centering
	\caption{System  Parameters}
	\label{tab1}
	\begin{tabular}{cc|cc}
		\hline
		{Parameter} &{Value} & Parameter & Value
		\\	\cline{1-4}	
		$M$	& 3 &$f$ &	2GHz
        		\\	
		$G_t$	& 20dB  &$G_r$ &	20dB
		\\
        $T_p$ & $10^{-6}$ &PRF &1000
        \\
        $B_w$ &100MHz &$B_c$ & 100MHz\\
        $
\left( x_s,y_s,z_s \right) 
$ &$
\left( 0,0,0 \right) 
$ &$\alpha$&$15^\circ$
        \\
        $\beta$ &$25^\circ$ &$P_{\mathrm{Rad}}^{\max}$&46dBm\\
        $P_{\mathrm{Com}}^{\max}$ &40dBm &$\mathrm{SNR}_{\min}$ &-10dB\\
		\hline
	\end{tabular}
	\label{tab2}
\end{table}

To the best of the authors’ knowledge, this work represents the first attempt to jointly investigate resource allocation and trajectory optimization for HAPs-enabled SAR systems. Due to the absence of established benchmark schemes in the existing literature, we construct two baseline approaches to facilitate comparative performance evaluation, in line with the practices adopted in related works \cite{lahmeri2022trajectory}:

\textbf{Benchmark 1}: The radar transmit power is fixed to the optimal value obtained from the proposed method. Under this setting, problem (P5) is re-solved to determine the corresponding benchmark performance.

\textbf{Benchmark 2}: The communication transmit power is fixed to the optimal value derived from the proposed algorithm. With this assumption, problem (P5) is re-optimized to obtain Benchmark 2.

\begin{figure}[htp]
	\centering
	\subfigure[]{
		{\includegraphics[width=0.38\textwidth]{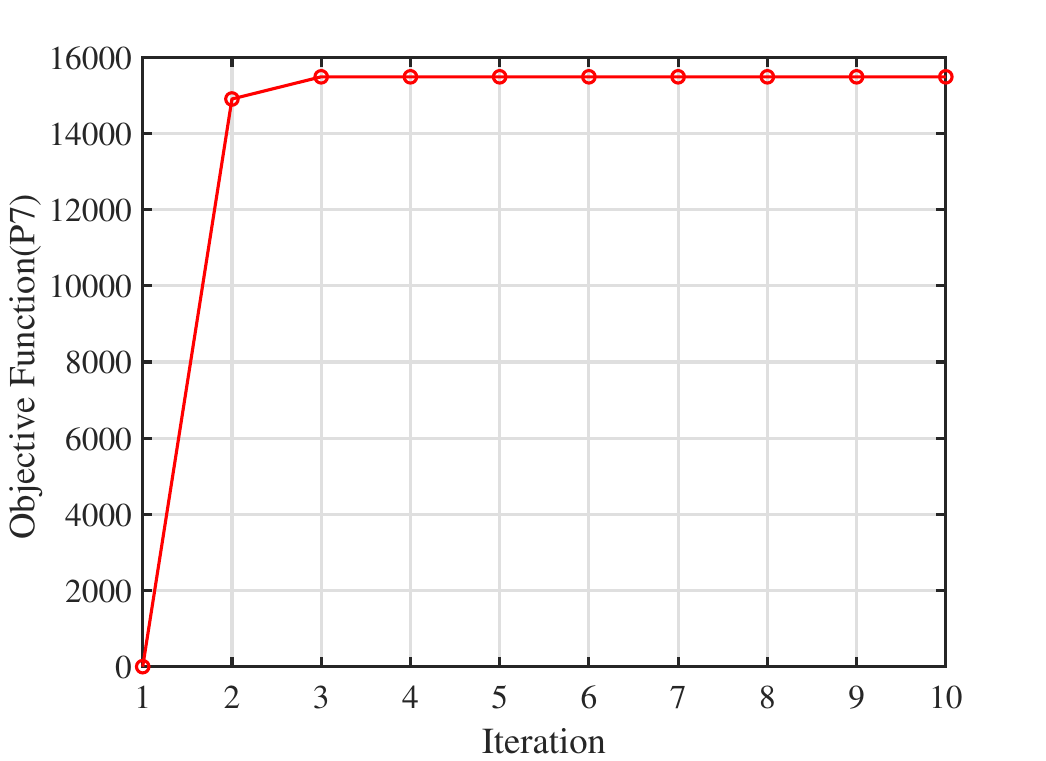}}}
	\subfigure[]{
		{\includegraphics[width=0.38\textwidth]{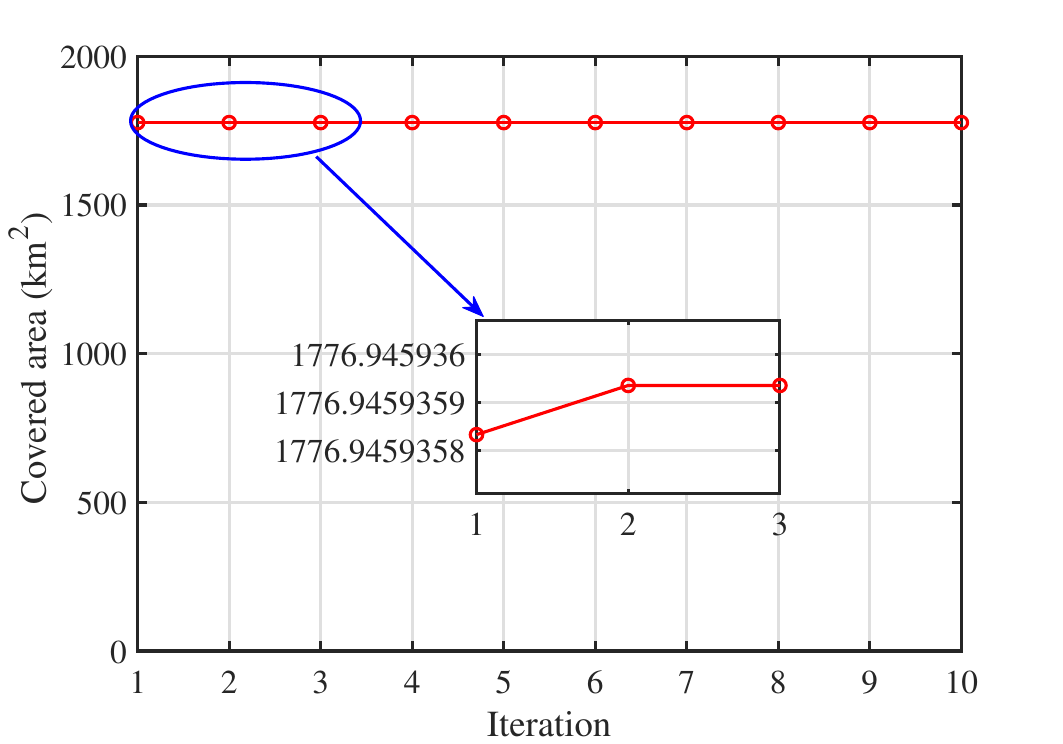}}}
	\caption{Convergence of algorithm \ref{alg1}. (a) Convergence of Problem (P7), (b) Convergence of Problem (P5).}
	\label{fig4}
\end{figure}
Fig. \ref{fig4} presents the convergence behavior of Algorithm 1, where Figs. \ref{fig4}(a) and 4(b) illustrate the convergence trajectories corresponding to optimization problems (P7) and (P5), respectively. Since the optimization over $N$ involves a finite search, it suffices to demonstrate the convergence of (P7) and (P5) to ensure the convergence of Algorithm \ref{alg1}. The results clearly indicate that the proposed algorithm converges rapidly, confirming the efficiency and effectiveness of the proposed optimization framework.

\begin{figure}[htp]
	\centering
	\subfigure{
		{\includegraphics[width=0.45\textwidth]{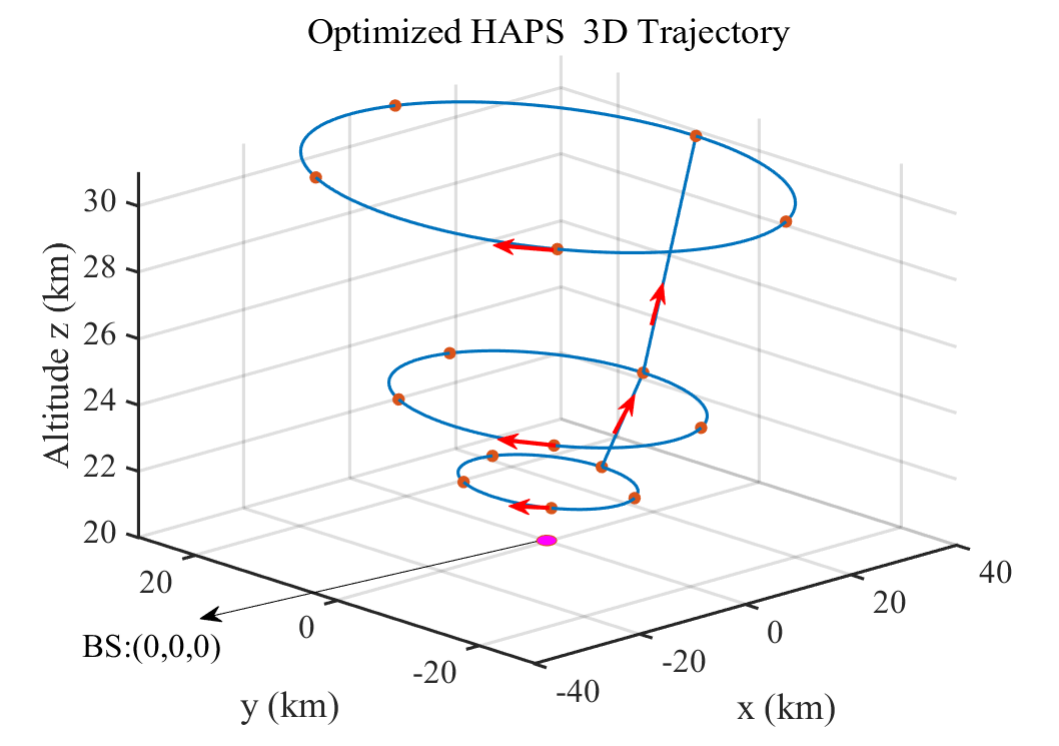}}}	
	\caption{Optimized HAPs-SAR 3D trajectory.}
	\label{fig5}
\end{figure}

Fig. \ref{fig5} illustrates the optimized three-dimensional trajectory of the HAPs, where the red arrows denote its flight direction and the magenta marker at the origin represents the ground BS location. The resulting trajectory configuration is governed by the trajectory constraints $C1$–$C4$, while the altitude variations arise from the sweep-wise adaptability permitted under constraint $C5$. From a system design perspective, although higher altitudes allow for broader beam coverage, they also incur significantly greater energy consumption. Despite the incorporation of solar energy harvesting via solar panels, the HAPs platform faces stringent energy limitations due to the substantial energy demands associated with sustained operation.
As a result, it is not feasible for the HAPs to continuously operate at maximum altitude throughout the mission. Therefore, an adaptive altitude strategy is essential, balancing between extended coverage area and energy efficiency. The derived trajectory reflects a trade-off that maximizes SAR coverage while ensuring operational sustainability under energy constraints.

\begin{figure}[htp]
	\centering
	{
		{\includegraphics[width=0.38\textwidth]{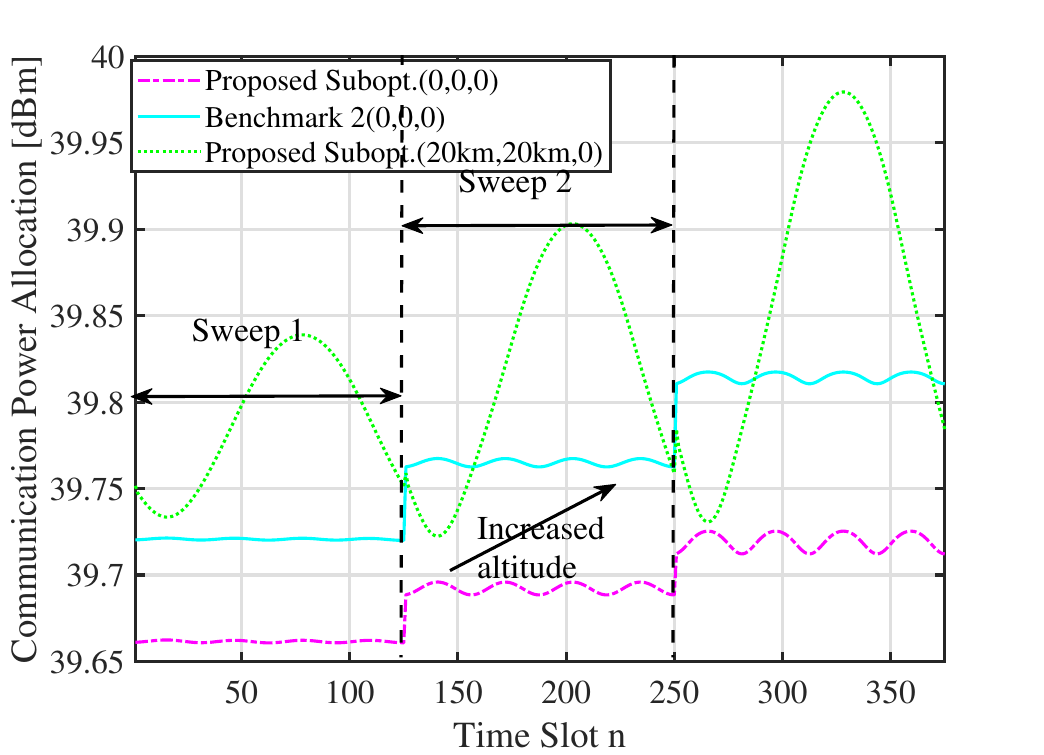}}}	
	\caption{Communication power allocation vs the time slot $n$.}
	\label{fig6}
\end{figure}

Fig. \ref{fig6} illustrates the communication power allocation under different BS deployment scenarios. The analysis reveals that as the HAPs ascends, increased communication power is necessary to satisfy constraint $C8$, which governs the minimum required data rate. This is attributed to the augmented path loss associated with higher altitudes, necessitating greater transmission power to maintain link quality. When the radar transmission power is fixed at its optimized value, the communication power can be further refined through optimization, as indicated by the blue solid line.
Furthermore, when the BS is positioned at (20km, 20km, 0), the communication power exhibits more pronounced temporal variations across time slots compared to the case where the BS is located at the origin (0, 0, 0), as shown by the green dashed line. This behavior arises from the dynamic distance changes between the HAPs and the BS due to the HAPs's circular trajectory. To maintain the required communication rate under varying distances, the HAPs must adapt its transmission power accordingly.
Additionally, the symmetry of the HAPs's circular motion imparts a corresponding symmetry to the communication power profile within a single sweep. This symmetry arises because the HAPs's distance to the BS varies in a predictable manner during its circular trajectory, leading to a mirrored pattern in the required communication power over time.
These observations underscore the critical importance of adaptive power control strategies in HAPs operations. By dynamically adjusting communication power in response to altitude changes and BS positioning, it is possible to enhance energy efficiency and maintain robust communication links throughout the mission.

\begin{figure}[htp]
	\centering
	{
		{\includegraphics[width=0.38\textwidth]{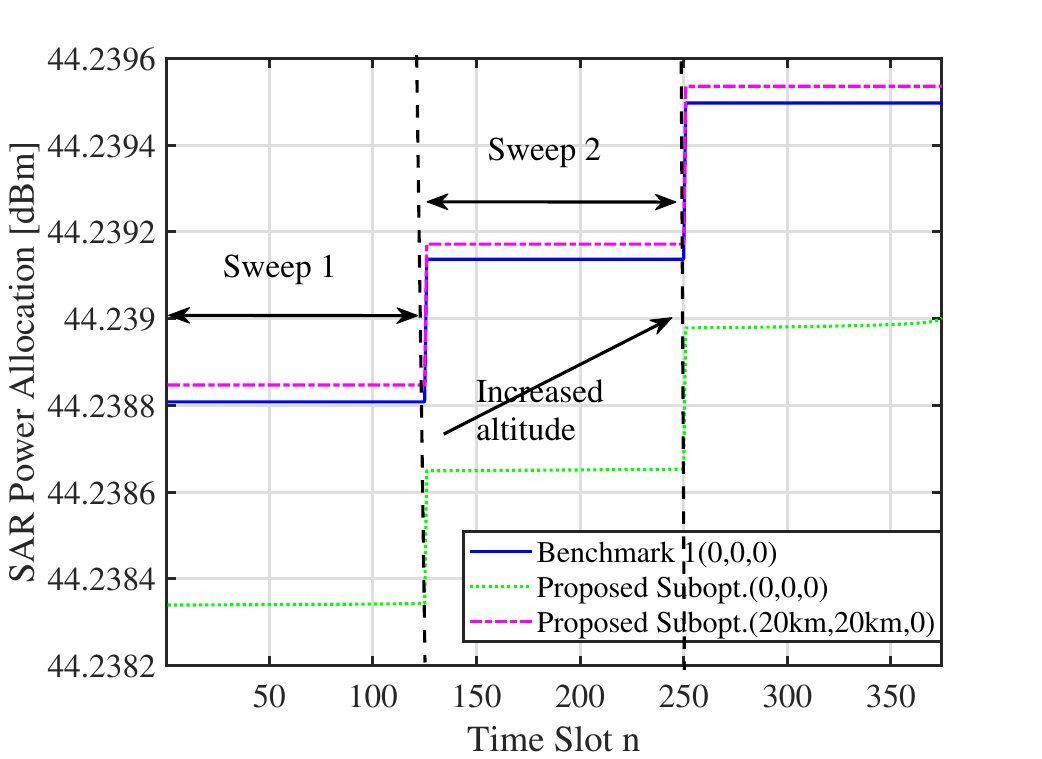}}}	
	\caption{SAR power allocation vs the time slot $n$.}
	\label{fig7}
\end{figure}

Fig. \ref{fig7} illustrates the SAR power allocation as a function of the time slot index. Unlike the communication power, which varies dynamically, the SAR transmission power remains constant within each sweep due to the enforcement of constraint $C10$. This constraint ensures that the radar system operates with uniform power across individual sweeps, maintaining consistency in SAR illumination.
Nonetheless, a similar trend is observed with respect to altitude. Namely as the HAPs ascends, the SAR system requires higher transmission power to ensure sufficient ground coverage and signal strength, owing to increased propagation losses at higher altitudes, which can be more clearly in \eqref{equ25}.
Furthermore, the location of the ground BS indirectly influences SAR power allocation. Variations in communication power, driven by the BS position and its associated rate constraints $C8$, can impose limitations on the total available power budget, thereby affecting the power allocated to SAR operations. This dependency is reflected by the magenta dashed line in the figure, highlighting the interplay between communication demands and radar resource allocation.

\begin{figure}[htp]
	\centering
	{
		{\includegraphics[width=0.38\textwidth]{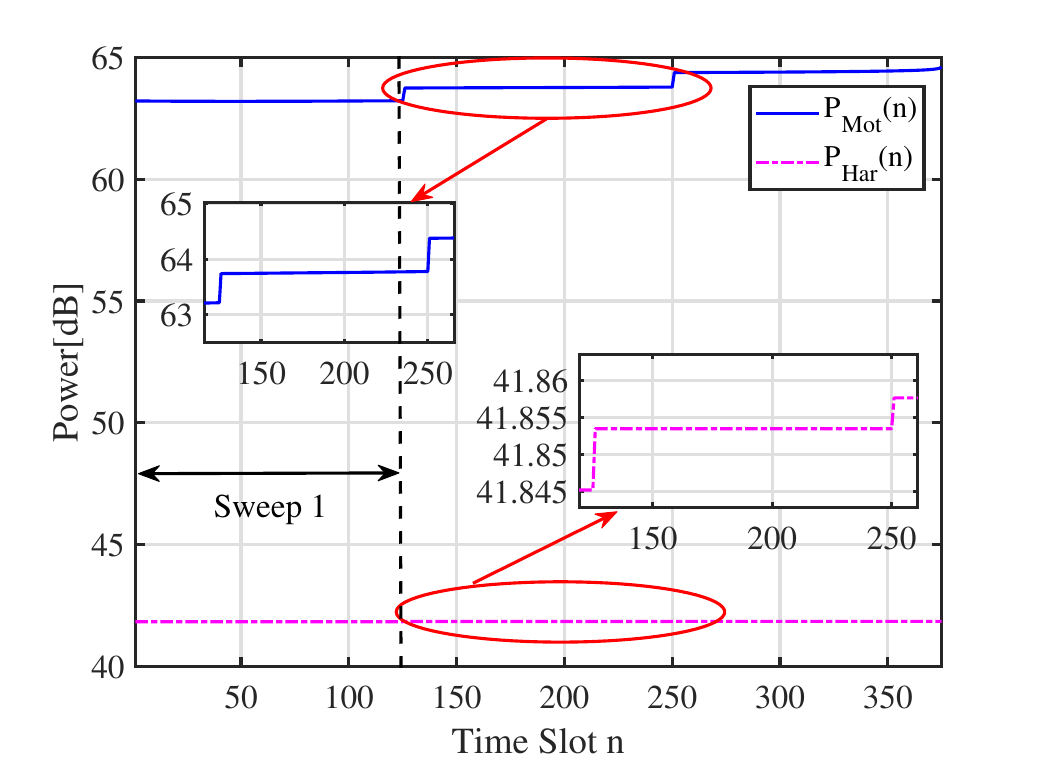}}}	
	\caption{Motion and harvesting power allocation.}
	\label{fig8}
\end{figure}
Fig. \ref{fig8} illustrates the temporal profiles of both the power consumed for maintaining HAPs mobility and the power harvested via onboard solar panels. Under the parameter configuration adopted in this paper, it is evident that the energy consumption consistently exceeds the harvested energy. This is primarily due to the analytical focus on the discharge phase of the solar-powered system, wherein the solar batteries are assumed to be fully charged, allowing for an isolated assessment of the HAPs-SAR system performance under energy depletion conditions.
Although this paper focuses on the energy-limited operation scenario, it is important to acknowledge that in practical deployments, the HAPs must also harvest sufficient energy to sustain prolonged operation. Future investigations may explore joint charging and operational strategies to ensure system sustainability, as discussed in \cite{javed2023interdisciplinary,Javed2025SystemDesign}.
Furthermore, it can be observed that within each sweep, the energy harvested and the energy consumed for propulsion are balanced. This is a direct consequence of constraints $C5$–$C7$, which enforce an energy neutrality condition over each sweep. Additionally, as the HAPs altitude increases, both the propulsion energy consumption and solar energy harvesting rise, driven respectively by increased aerodynamic resistance and enhanced solar irradiance at higher altitudes.

\begin{figure}[htp]
	\centering
	{
		{\includegraphics[width=0.38\textwidth]{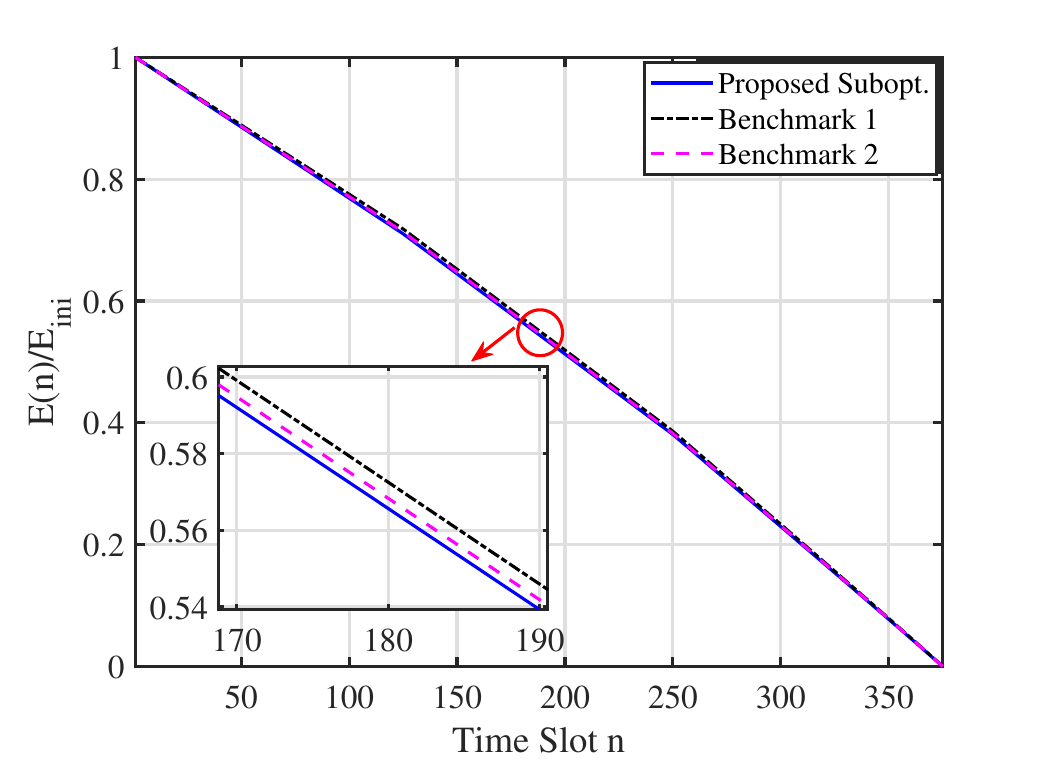}}}	
	\caption{Ratio of remaining energy to initial energy over time for the HAPs.}
	\label{fig9}
\end{figure}

Fig. \ref{fig9} presents the ratio of the remaining energy to the initial energy of the HAPs at each time slot $n$. Although the remaining energy reaches zero upon completion of the prescribed sensing mission, which is mathematically permissible, this condition is impractical for real-world systems. In practice, the HAPs must retain sufficient energy to either return to base or remain in standby mode for future tasks. Nevertheless, since this paper focuses solely on the discharging process and only requires satisfaction of constraint $C14$, the results depicted in Fig. \ref{fig9} are considered feasible within the context of the proposed model. Furthermore, it is evident that the energy ratio decreases approximately linearly over time, with proposed method exhibiting the most rapid energy consumption among the compared schemes.

\begin{figure}[htp]
	\centering
	\subfigure{
	{\includegraphics[width=0.28\textwidth]{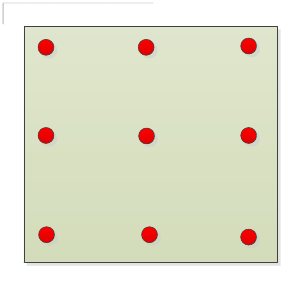}}}	
	\caption{SAR imaging scenario with a 9-target lattice configuration.}
	\label{fig10}
\end{figure}

\begin{figure}[htp]
	\centering
	\subfigure{
		{\includegraphics[width=0.45\textwidth]{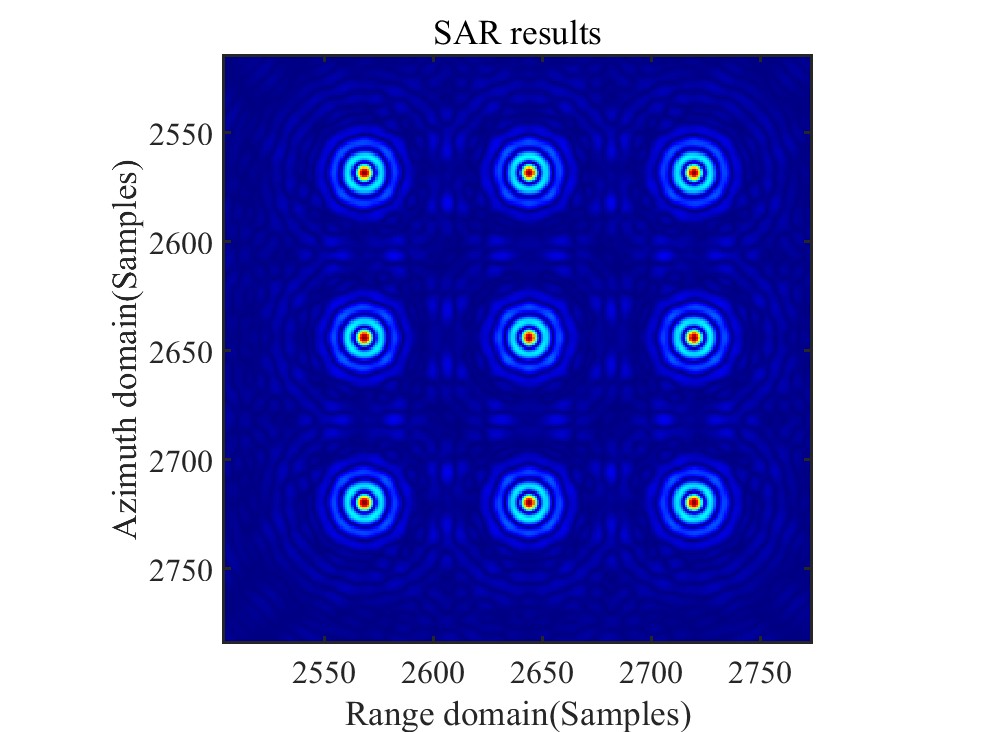}}}	
	\caption{SAR imaging results corresponding to the 9-target lattice configuration.}
	\label{fig11}
\end{figure}

To further verify the feasibility of the proposed trajectory and resource allocation strategy for SAR imaging, a simulation was conducted on the 9-target lattice scenario depicted in Fig. \ref{fig10}, using the parameters outlined in this paper and a well-established polar formatting algorithm (PFA) \cite{ZhangLin2023HolographicSAR,LinHong2011ExtensionofRange,ZhangSun2025LargeScene}. The corresponding imaging results are shown in Fig. \ref{fig11}. It can be clearly seen that all targets, both in the center of the scene and in the periphery, are accurately focused. These results provide additional evidence supporting the effectiveness of the proposed optimization framework.

\section{Conclusion}
\label{sec5}

In this work, we proposed a joint trajectory and resource optimization framework for HAPs-SAR systems aimed at real-time sensing and energy efficiency. By offloading SAR image reconstruction to a ground BS, the HAPs can operate with reduced onboard processing and energy costs. We developed a detailed flight model and performed energy analysis considering solar power harvesting and consumption from radar sensing, communication, and motion. A MINLP problem was formulated to maximize beam coverage under system constraints, and an SCA-based sub-optimal solution was proposed. Extensive simulations verified the convergence and effectiveness of the proposed algorithm across multiple performance dimensions. The SAR imaging results for a 9-target lattice scenario further validated the practicality and robustness of the optimization framework. Future work may extend this framework to include adaptive scheduling, multi-HAPs cooperative scenarios, and validation using real-world measurement data.


\bibliographystyle{IEEEtran}
\bibliography{ref}
\end{document}